\documentclass[a4paper,UKenglish,cleveref, autoref, thm-restate]{lipics-v2021}
\newif\ifiscameraready
\iscamerareadyfalse

\ifiscameraready
\else
\hideLIPIcs  %
\fi

\title{On the Finite Variable-Occurrence Fragment of the Calculus of Relations with Bounded Dot-Dagger Alternation} %

\titlerunning{On the Finite Variable-Occurrence Fragment of the Calculus of Relations ...} %

\author{Yoshiki {Nakamura}}{Tokyo Institute of Technology, Japan}{nakamura.yoshiki.ny@gmail.com}{https://orcid.org/0000-0003-4106-0408}{}%

\authorrunning{Y. Nakamura} %

\Copyright{Yoshiki Nakamura} %



\category{} %

\ifiscameraready
\relatedversion{} %
\relatedversiondetails[cite=nakamuraFiniteVariableOccurrence2023arxiv]{Extended Version}{} %
\fi
\supplement{}%

\funding{This work was supported by JSPS KAKENHI Grant Number JP21K13828.}%

\acknowledgements{We would like to thank the anonymous reviewers for their useful comments.}%

\nolinenumbers 

\EventEditors{J\'{e}r\^{o}me Leroux, Sylvain Lombardy, and David Peleg}
\EventNoEds{3}
\EventLongTitle{48th International Symposium on Mathematical Foundations of Computer Science (MFCS 2023)}
\EventShortTitle{MFCS 2023}
\EventAcronym{MFCS}
\EventYear{2023}
\EventDate{August 28 to September 1, 2023}
\EventLocation{Bordeaux, France}
\EventLogo{}
\SeriesVolume{272}
\ArticleNo{69}
\usepackage{listings}
\makeatletter
\def\lst@lettertrue{\let\lst@ifletter\iffalse}
\makeatother

\usepackage{amsmath}
\usepackage{mathtools}

\usepackage{graphviz}
\usepackage{svg}

\Crefname{theorem}{Thm.}{Thms.}
\Crefname{proposition}{Prop.}{Props.}
\Crefname{lemma}{Lem.}{Lems.}
\Crefname{corollary}{Cor.}{Cors.}
\Crefname{definition}{Def.}{Defs.}
\crefname{section}{Sect.}{Sects.}
\crefname{figure}{Fig.}{Figs.}

\usepackage{tikz}
\usetikzlibrary{calc, trees, bending, backgrounds, shapes, shapes.geometric, fit, arrows, arrows.meta, positioning, graphs}
\tikzset{earrow/.style={>={{[flex] Latex[length=.08cm, width=2.5pt]}}}}
\usepackage{algorithm, algpseudocode}

\usepackage{xcolor}
\definecolor[named]{ACMBlue}{cmyk}{1,0.1,0,0.1}
\definecolor[named]{ACMYellow}{cmyk}{0,0.16,1,0}
\definecolor[named]{ACMOrange}{cmyk}{0,0.42,1,0.01}
\definecolor[named]{ACMRed}{cmyk}{0,0.90,0.86,0}
\definecolor[named]{ACMLightBlue}{cmyk}{0.49,0.01,0,0}
\definecolor[named]{ACMGreen}{cmyk}{0.20,0,1,0.19}
\definecolor[named]{ACMPurple}{cmyk}{0.55,1,0,0.15}
\definecolor[named]{ACMDarkBlue}{cmyk}{1,0.58,0,0.21}
\hypersetup{colorlinks,
    linkcolor=ACMPurple,
    citecolor=ACMPurple,
    urlcolor=ACMDarkBlue,
    filecolor=ACMDarkBlue}
\usepackage{hyperref}

\usepackage{rotating}

\usepackage{diagbox}

\MakeRobust{\ref}%

\makeatletter
\newcommand{\labeltext}[2]{%
    \@bsphack
    \csname phantomsection\endcsname %
    \def\@currentlabel{#1}{\label{#2}}%
    \@esphack
}
\makeatother

\usepackage[xcolor, notion]{knowledge}
\knowledgeconfigureenvironment{theorem,lemma,proof}{}

\usepackage{stackengine}
\newcommand{\defeq}{\mathrel{\ensurestackMath{\stackon[1pt]{=}{\scriptscriptstyle\Delta}}}}
\newcommand{\defiff}{\mathrel{\ensurestackMath{\stackon[1pt]{\iff}{\scriptscriptstyle\Delta}}}}
\usepackage{scalerel,stackengine}
\stackMath
\newcommand\reallywidehat[1]{%
    \savestack{\tmpbox}{\stretchto{%
            \scaleto{%
                \scalerel*[\widthof{\ensuremath{#1}}]{\kern-.6pt\bigwedge\kern-.6pt}%
                {\rule[-\textheight/2]{1ex}{\textheight}}%
            }{\textheight}%
        }{0.5ex}}%
    \stackon[1pt]{#1}{\tmpbox}%
}
\knowledgenewrobustcmd\tuple[1]{\langle #1 \rangle}
\knowledgenewrobustcmd\set[1]{\{ #1 \}}
\knowledgenewrobustcmd\card{\mathop{\cmdkl{\#}}}
\knowledgenewrobustcmd\powerset{\mathop{\cmdkl{\wp}}}
\knowledgenewrobustcmd\range[1]{\cmdkl{[} #1 \cmdkl{]}}

\knowledgenewrobustcmd\quoset{\cmdkl{/}}
\knowledgenewrobustcmd\quoclass[1]{\cmdkl{[}#1\cmdkl{]}}

\usepackage{stmaryrd}
\knowledgenewrobustcmd\jump[1]{\cmdkl{\llbracket} #1 \cmdkl{\rrbracket}}

\NewDocumentCommand\const{m}{\mathsf{#1}}

\knowledgenewrobustcmd{\nat}{\cmdkl{\mathbb{N}}}
\knowledgenewrobustcmd{\pnat}{\cmdkl{\mathbb{N}}_{+}}

\usepackage{ebproof}

\knowledgenewrobustcmd\sig{\cmdkl{S}}
\knowledgenewrobustcmd\vsig{\cmdkl{\Sigma}}

\NewDocumentCommand\term{O{1}}{%
    \ifcase#1
        undefined
    \or t
    \or s
    \or {u}
    \else undefined

    \fi
}
\NewDocumentCommand\termset{O{1}}{%
    \ifcase#1
        undefined
    \or T
    \else undefined

    \fi
}
\knowledgenewrobustcmd\allterm{\cmdkl{\mathbf{T}}}

\NewDocumentCommand\fml{O{1}}{%
    \ifcase#1
        undefined
    \or \varphi
    \or \psi
    \or \rho
    \else undefined

    \fi
}

\knowledgenewrobustcmd\assign[2]{\cmdkl{[}#1\cmdkl{/}#2\cmdkl{]}}

\knowledgenewrobustcmd\vo{\cmdkl{\mathop{\mathrm{vo}}}}

\newcommand{\val}{\mathfrak{v}}
\knowledgenewrobustcmd\voset[2]{#1_{\cmdkl{(\vo \le #2)}}}
\knowledgenewrobustcmd\voseteq[2]{#1_{\cmdkl{(\vo = #2)}}}
\newcommand{\algclass}{\mathcal{C}}

\knowledgenewrobustcmd\simalgclass[1]{\mathbin{\cmdkl{\sim}_{#1}}}

\knowledgenewrobustcmd{\extpvsig}{\cmdkl{\dot{\Sigma}}}

\newrobustcmd{\extpvsigzero}{\kl[\extpvsigzero]{\dot{\Sigma}}_{\kl[\extpvsigzero]{0}}}
\knowledge{\extpvsigzero}{automatic in command}

\knowledgenewrobustcmd{\extpsim}[1]{\mathbin{\cmdkl{\dot{\sim}}}_{#1}}

\NewDocumentCommand\word{O{1}}{%
    \ifcase#1
        undefined
    \or w
    \or v
    \else undefined

    \fi
}

\newcommand{\CoR}{\mathrm{CoR}}
\newrobustcmd\CoRallterm{\kl[\CoRallterm]{\mathbf{T}}^{\kl[\CoRallterm]{\CoR}}}
\knowledge{\CoRallterm}{automatic in command}
\newrobustcmd\CoRsig{\kl[\CoRsig]{S}^{\kl[\CoRsig]{\CoR}}}
\knowledge{\CoRsig}{automatic in command}
\newrobustcmd{\CoRSigma}{\kl[\CoRSigma]{\Sigma}^{\kl[\CoRSigma]{\CoR}}}
\knowledge{\CoRSigma}{automatic in command}
\newrobustcmd{\CoRPi}{\kl[\CoRPi]{\Pi}^{\kl[\CoRPi]{\CoR}}}
\knowledge{\CoRPi}{automatic in command}

\NewDocumentCommand\model{O{1}}{%
    \ifcase#1
        undefined
    \or \mathfrak{A}
    \or \mathfrak{B}
    \else undefined

    \fi
}

\knowledgenewrobustcmd{\REL}{\cmdkl{\mathsf{REL}}}

\knowledge{notion}
| KMW
| KMW class

\knowledge{notion}
| BSR
| BSR class

\knowledge{notion}
| subterm-closed

\knowledge{notion}
| words
| word

\knowledge{notion}
| monoids
| monoid

\knowledge{notion}
| $\sig$-algebra
| $\sig$-algebras

\knowledge{notion}
| equational theory
| equational theories

\knowledge{notion}
| $k$-variable fragment
| $1$-variable fragment
| $2$-variable fragment

\knowledge{notion}
| $k$-variable-occurrence fragment
| $1$-variable-occurrence fragment
| $4$-variable-occurrence fragment

\knowledge{notion}
| CoR
| CoR terms

\knowledge{notion}
| bounded dot-dagger alternation
| dot-dagger alternation

\knowledge{notion}
| dot-dagger alternation hierarchy

\knowledge{notion}
| structure
| structures

\knowledge{notion}
| cofinite
| cofiniteness

\begin{document}

\maketitle

\begin{abstract}
We introduce the \emph{$k$-variable-occurrence fragment}, which is the set of terms having at most $k$ occurrences of variables.
We give a sufficient condition for the decidability of the equational theory of the $k$-variable-occurrence fragment using the finiteness of a monoid.
As a case study, we prove that for Tarski's calculus of relations with bounded dot-dagger alternation (an analogy of quantifier alternation in first-order logic), the equational theory of the $k$-variable-occurrence fragment is decidable for each $k$.
 \end{abstract}

\section{Introduction}\label{section: introduction}
Since \emph{the satisfiability problem of first-order logic} is undecidable \cite{Church36, Turing37} in general,
(un-)decidable classes of first-order logic are widely studied \cite{Borger1997};
for example, the undecidability holds even for the \intro[KMW]{Kahr--Moore--Wang (KMW) class}\footnote{Recall the notation for \emph{prefix-vocabulary classes} \cite[Def.\ 1.3.1]{Borger1997}.
    E.g., $[\forall\exists\forall, (0, \omega)]$ denotes the set of prenex sentences $\fml$ of first-order logic without equality, function symbols, nor constants such that
    the quantifier prenex of $\fml$ is $\forall \exists \forall$;
    $\fml$ has $\omega$ (countably infinitely many) binary relation symbols and $\fml$ does not have 1- nor $i$-ary relation symbols for $i \ge 3$.} $[\forall\exists\forall, (0, \omega)]$ \cite{Kahr62},
but it is decidable for the \intro[BSR]{Bernays--Sch{\"o}nfinkel--Ramsey (BSR) class} $[\exists^*\forall^*, \mathrm{all}]_{=}$ \cite{Bernays1928,Ramsey1930}.
\kl[CoR]{The calculus of relations (CoR)} \cite{Tarski1941}, revived by Tarski, is an algebraic system on binary relations; its expressive power is equivalent to that of the three-variable fragment of first-order logic with equality \cite{Tarski1941, Tarski1987}, w.r.t.\  binary relations.
The equational theory of \kl{CoR} is undecidable \cite{Tarski1941, Tarski1987}\footnote{In \cite{Tarski1987}, the undecidability of the \kl{equational theory} is shown for more general classes of relation algebras.} in general, which follows from the undecidability of the \kl{KMW class},
but, for example, it is decidable for the \emph{(existential) positive fragment} \cite{Andreka1995, Pous2018} and the \emph{existential fragment} \cite{nakamuraExistentialCalculiRelations2023} of \kl{CoR}, which follows from the decidability of the \kl{BSR class}.
On the undecidability of \kl{CoR}, the undecidability holds even for the \kl{$1$-variable fragment} \cite{Maddux94} and even for the \kl{$1$-variable fragment} only with union, composition, and complement \cite{Nakamura2019}, where the \intro{$k$-variable fragment} denotes the set of terms having at most $k$ variables.

Then, from the undecidability result for the \kl{$1$-variable fragment} of \kl{CoR} \cite{Maddux94,Nakamura2019} above,
the following natural question arises---Is it decidable for the \kl{$k$-variable-occurrence fragment} of \kl{CoR}?
Here, the \kl{$k$-variable-occurrence fragment} denotes the set of terms having at most $k$ occurrences of variables.
For example, when $a, b$ are variables and $\const{I}$ is a constant,
the term $(a \cdot b) \cdot (\const{I} \cdot (a \cdot b))$ has $4$ occurrences of variables
and $1$ occurrence of constants;
thus, this term is in the \kl{$4$-variable-occurrence fragment} (cf.\ the term is in the \kl{$2$-variable fragment} since the variables $a, b$ occur).
While one may seem that this restriction immediately implies the decidability, the equational theory of the \kl{$k$-variable-occurrence fragment} on some (single) algebra is undecidable in general even when $k = 0$ (\Cref{remark: undecidable}).

Our contribution is to prove that the equational theory of the \kl{$k$-variable-occurrence fragment} is decidable for \kl{CoR} \emph{with \kl{bounded dot-dagger alternation}},
where the \kl{dot-dagger alternation} \cite{nakamuraExpressivePowerSuccinctness2020, nakamuraExpressivePowerSuccinctness2022} is an analogy of the quantifier alternation in first-order logic.
Note that the equational theory of the \kl{$k$-variable fragment} is undecidable in general for \kl{CoR} \cite{Maddux94, Nakamura2019} (even with \kl{bounded dot-dagger alternation} (\Cref{proposition: decidable/undecidable})).
Our strategy is to prove that the number of terms in the \kl{$k$-variable-occurrence fragment} is finite up to the semantic equivalence relation.
To this end, \labeltext{(1)}{outline:1}(1) we decompose terms as much as possible; and then \labeltext{(2)}{outline:2}(2) we show that each decomposed part is finite up to the semantic equivalence relation by collecting valid equations.
By the preprocessing of \ref{outline:1}, one can see that for \ref{outline:2}, essentially, it suffices to prove the finiteness of some monoid (using the method of \Cref{section: 1-variable}).
Its finiteness is not clear, as it is undecidable whether a (finitely presented) monoid is finite in general;
but, fortunately, we can prove the finiteness (\Cref{theorem: main}) by finding valid equations (\Cref{figure: equations}).

The rest of this paper is structured as follows.
\Cref{section: 1-variable} introduces the \kl{$k$-variable-occurrence fragment} for general algebras and gives a framework to prove the decidability from the finiteness of a monoid.
\Cref{section: CoR} recalls the syntax and semantics of \kl{CoR} and the \kl{dot-dagger alternation hierarchy}.
In \Cref{section: CoR k}, based on \Cref{section: 1-variable}, we prove that the equational theory of \kl{CoR} with \kl{bounded dot-dagger alternation} is decidable.
\Cref{section: conclusion} concludes this paper.

We write $\intro*\nat$ for the set of all non-negative integers.
For $l, r \in \nat$, we write $\intro*\range{l, r}$ for the set $\set{i \in \nat \mid l \le i \le r}$.
For a set $A$, we write $\intro*\card A$ for the cardinality of $A$ and $\intro*\powerset(A)$ for the power set of $A$.
For a set $A$ and an equivalence relation $\sim$ on $A$, we write $A\intro*\quoset{\sim}$ for the quotient set of $A$ by $\sim$ and $\intro*\quoclass{a}_{\sim}$ for the equivalence class of an element $a \in A$ on $\sim$.
\knowledgenewrobustcmd{\univ}[1]{\cmdkl{|}#1\cmdkl{|}}

\section{On the $k$-variable-occurrence fragment}\label{section: 1-variable}
We fix $\intro*\vsig$ as a non-empty \emph{finite} set of variables.
We fix $\intro*\sig$ as a \emph{finite} algebraic signature; $\sig$ is a map from a finite domain (of functions) to $\nat$.
For each $\tuple{f, n} \in \sig$, we write $f^{(n)}$; it is the function symbol $f$ with arity $n$.
We also let $\sig^{(n)} \defeq \set{f^{(m)} \in \sig \mid m = n}$.
The set $\intro*\allterm$ of \emph{$\sig$-terms} over $\vsig$ is defined as the minimal set closed under the following two rules:
$a \in \vsig \Longrightarrow a \in \reintro*\allterm$; ($f^{(n)} \in \sig$ and $\term_1, \dots, \term_n \in \allterm$) $\Longrightarrow$ $f(\term_1, \dots, \term_n) \in \reintro*\allterm$.

An \intro{$\sig$-algebra} $A$ is a tuple $\tuple{\intro*\univ{A}, \set{f^{A}}_{f^{(n)} \in \sig}}$, where $\univ{A}$ is a non-empty finite set and $f^{A} \colon \univ{A}^n \to \univ{A}$ is an $n$-ary map for each $f^{(n)} \in \sig$.
A \emph{valuation} $\val \colon \vsig \to \univ{A}$ on an \kl{$\sig$-algebra} $A$ is a map;
we write $\hat{\val} \colon \allterm \to \univ{A}$ for the unique homomorphism extending $\val$.
For a class $\algclass$ of \kl{$\sig$-algebras},
the equivalence relation $\intro*\simalgclass{\algclass}$ on $\allterm$ is defined by:
$\term[1] \intro*\simalgclass{\algclass} \term[2] \defiff \mbox{$\hat{\val}(\term[1]) = \hat{\val}(\term[2])$ for all valuations $\val$ on all algebras in $\algclass$}$.
For a set $\termset \subseteq \allterm$, the \intro[equational theory]{equational theory of $\termset$ over $\algclass$} is the set $\set{\tuple{\term[1], \term[2]} \in \termset^2 \mid \term[1] \simalgclass{\algclass} \term[2]}$.

\begin{definition}[$k$-variable-occurrence fragment]\label{definition: k-vo fragment}
    For an $\sig$-term $\term$, let $\intro*\vo(\term)$ be the number of occurrences of variables in $\term$:
    $\intro*\vo(\term) \defeq \begin{cases}
            1                             & (\term \in \vsig)                    \\
            \sum_{i = 1}^{n} \vo(\term_i) & (\term = f(\term_1, \dots, \term_n))
        \end{cases}$.
    For each set $\termset$ of $\sig$-terms,
    the \intro{$k$-variable-occurrence fragment} $\intro*\voset{\termset}{k}$ is the set $\set{\term \in \termset \mid \vo(\term) \le k}$.
    (Similarly, let $\intro*\voseteq{\termset}{k} \defeq \set{\term \in \termset \mid \vo(\term) = k}$.)
    Clearly, $\termset = \bigcup_{k \in \nat} \voset{\termset}{k}$.
\end{definition}
\begin{remark}\label{remark: undecidable}
    The \kl{equational theory} of the \kl{$k$-variable-occurrence fragment} is undecidable in general, even when $k = 0$.
    It follows from the reduction from the word problem for \kl{monoids}.
    Let $M = \tuple{\univ{M}, \circ^{M}, \const{I}^{M}}$ be a (finitely) presented \kl{monoid} with finite generators $C = \set{c_1, \dots, c_l}$ such that the word problem for $M$ is undecidable (by Markov \cite{markovImpossibilityCertainAlgorithms1947} and Post \cite{postRecursiveUnsolvabilityProblem1947}).
    We define $\sig \defeq \set{c_1^{(1)}, \dots, c_l^{(1)}} \cup \set{\const{I}^{(0)}}$ and the \kl{$\sig$-algebra} $A \defeq \tuple{\univ{A}, \set{f^{A}}_{f^{(n)} \in \sig}}$ by:
    $\univ{A} = \univ{M}$;
    $c_i^{A}(x) = c_i \circ^{M} x$ for $i \in \range{1, l}$;
    $\const{I}^{A} = \const{I}^{M}$.
    By definition, for all two words $a_1 \dots a_n$, $b_1 \dots b_m$, over $C$:
    they are equivalent in $M$ iff
    $a_1(a_2(\dots a_{n}(\const{I}) \dots)) \simalgclass{\set{A}} b_1(b_2(\dots b_{m}(\const{I}) \dots))$.
\end{remark}
In the rest of this section, we fix $\algclass$ as a class of \kl{$\sig$-algebras}.
\subsection{On the finiteness of $k$-variable-occurrence fragment: from $1$ to $k$}
How can we show the decidability of the \kl{equational theory} of the \kl{$k$-variable-occurrence fragment}?
We consider proving it from the finiteness up to the semantic equivalence relation:
\begin{proposition}[Cor.\ of {\cite{bussBooleanFormulaValue1987, lohreyParallelComplexityTree2001}} for the complexity]\label{proposition: finite to decidable}
    Let $\termset \subseteq \allterm$ be a \kl{subterm-closed}\footnote{A set $\termset \subseteq \allterm$ is \intro{subterm-closed} if for every $\term \in \termset$, if $\term[2]$ a subterm of $\term$, then $\term[2] \in \termset$.} set.
    If the set $\termset\quoset{\simalgclass{\algclass}}$ is finite,
    the \kl{equational theory} of $\termset$ over $\algclass$ is decidable.
    Moreover, it is decidable in DLOGTIME-uniform $\mathrm{NC^{1}}$ if the input is given as a well-bracketed string.
\end{proposition}
\begin{proof}[Proof Sketch]
    Because $\algclass$ is fixed and $\termset\quoset{\simalgclass{\algclass}}$ is finite, for each $\term \in \termset$, one can calculate the index of the equivalence class of $\term$ on $\simalgclass{\algclass}$ by using the (finite and possibly partial) Cayley table of each operator; thus, the equational theory is decidable.
    Moreover, according to this algorithm, if the input is given as a well-bracketed string, one can also construct a parenthesis context-free grammar such that for all $\term[1], \term[2] \in \termset$,
    the well-bracketed string encoding the equation $\term[1] = \term[2]$ is in the language iff $\term[1] \simalgclass{\algclass} \term[2]$.
    Hence, the complexity is shown because every language recognized by a parenthesis context-free grammar is in ALOGTIME \cite{bussBooleanFormulaValue1987, Buss92} (ALOGTIME is equivalent to DLOGTIME-uniform $\mathrm{NC^{1}}$ \cite{Barrington90}).
\end{proof}
For the \kl{$k$-variable-occurrence fragment}, the finiteness of $\voset{\termset}{1}$ (with \Cref{proposition: finite to decidable}) can imply the decidability of the \kl{equational theory} of $\voset{\termset}{k}$ (\Cref{lemma: 1 to k}) by the following decomposition lemma.
Here, we write $\term\intro*\assign{\term[2]}{a}$ for the term $\term$ in which each $a$ has been replaced with $\term[2]$.
\begin{lemma}\label{lemma: decomposition}
    Let $\termset \subseteq \allterm$ be a \kl{subterm-closed} set.
    Let $k \ge 2$, $a \in \vsig$.
    Then, for all $\term \in \voseteq{\termset}{k}$, there are $\term_0 \in \voset{\termset}{1}, f^{(n)} \in \sig, \term_1, \dots, \term_n \in \voset{\termset}{k-1}$
    such that $\term = \term_0\assign{f(\term_1, \dots, \term_n)}{a}$.
\end{lemma}
\begin{proof}
    By induction on $\term$.
    Since $k \ge 2$,
    there are $g^{(m)} \in \sig, \term[2]_1, \dots, \term[2]_m \in \termset$ s.t.\ $\term = g(\term[2]_1, \dots, \term[2]_m)$
    and $\sum_{i = 1}^{m} \vo(\term[2]_i) = k$.
    Case $\vo(\term[2]_i) \le k - 1$ for all $i$:
    By letting $\term_0 \defeq a$, we have $\term = \term_0 \assign{g(\term[2]_1, \dots, \term[2]_m)}{a}$.
    Otherwise:
    Let $i$ be s.t.\ $\vo(\term[2]_i) = k$.
    Since $\term[2]_i \in \voseteq{\termset}{k}$,
    let $\term[3] \in \voset{\termset}{1}, f^{(n)} \in \sig, \term_1, \dots, \term_n \in \voset{\termset}{k-1}$
    be the ones obtained by IH w.r.t.\ $\term[2]_i$, so that $\term[2]_i = \term[3]\assign{f(\term_1, \dots, \term_n)}{a}$.
    By letting $\term_0 \defeq g(\term[2]_1, \dots, \term[2]_{i-1}, \term[3], \term[2]_{i+1}, \dots, \term[2]_m)$,
    we have:
    \begin{align*}
        \term_0\assign{f(\term_1, \dots, \term_n)}{a}
         & = g(\term[2]_1, \dots, \term[2]_{i-1}, \term[3], \term[2]_{i+1}, \dots, \term[2]_m)\assign{f(\term_1, \dots, \term_n)}{a}                                           \\
         & = g(\term[2]_1, \dots, \term[2]_{i-1}, \term[3]\assign{f(\term_1, \dots, \term_n)}{a}, \term[2]_{i+1}, \dots, \term[2]_m) \tag{$\vo(\term[2]_j) = 0$ if $j \neq i$} \\
         & = g(\term[2]_1, \dots, \term[2]_{i-1}, \term[2]_i, \term[2]_{i+1}, \dots, \term[2]_m) = \term. \tag{$\term[2]_i = \term[3]\assign{f(\term_1, \dots, \term_n)}{a}$}
    \end{align*}
    Hence, this completes the proof.
\end{proof}
\begin{example}[of \Cref{lemma: decomposition}] %
    If $\sig = \set{\circ^{(2)}, \const{I}^{(0)}}$ and $a, b, c \in \vsig$,
    the term $\term[1] = \const{I} \circ ((a \circ (b \circ c)) \circ \const{I}) \in \voset{\allterm}{3}$ has the following decomposition:
    $\term[1] = (\const{I} \circ (a \circ \const{I}))\assign{a \circ (b \circ c)}{a}$.
    Then $\const{I} \circ (a \circ \const{I}) \in \voset{\allterm}{1}$ and $a, (b \circ c) \in \voset{\allterm}{2}$.
    The following is an illustration of the decomposition, where the number written in each subterm $\term[2]$ denotes $\vo(\term[2])$:
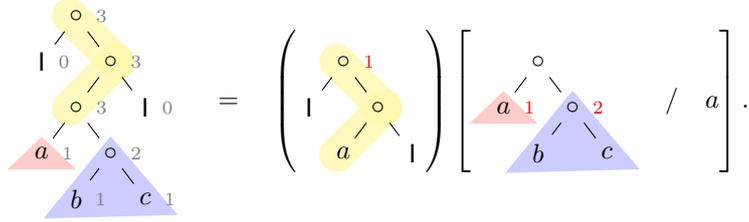
    \[
        \begin{tikzpicture}[
                level distance=4ex,
                sibling distance=6ex,
                baseline=-8ex
            ]
            \tikzset{fit1/.style={rounded rectangle, inner sep=0.8pt, fill=yellow!40, opacity = .8},
                fit2/.style={red!20, fill = red!20, thick},
                fit3/.style={blue!20, fill = blue!20, thick},
                nlab/.style={font= \scriptsize, color = gray}}
            \node(0){$\circ$}
            child {node(00){$\const{I}$}}
            child {node(01){$\circ$}
                    child {node(010) {$\circ$}
                            child {node(0100) {$a$}
                                }
                            child {node(0101) {$\circ$}
                                    child {node(01010) {$b$}
                                        }
                                    child {node(01011) {$c$}
                                        }
                                }
                        }
                    child {node(011) {$\const{I}$}
                        }
                }
            ;
            \begin{pgfonlayer}{background}
                \node[fit1, rotate fit=-45, fit=(0)(01)] {};
                \node[fit1, rotate fit=45, fit=(01)(010)] {};
                \draw [fit2] ($(0100.north) + (0,0)$) -- ($(0100.south west) + (-.2, 0)$) -- ($(0100.south east) + (.2, 0)$) -- cycle;
                \draw [fit3] ($(0101.north) + (0,0)$) -- ($(01010.south west) + (-.2, 0)$) -- ($(01011.south east) + (.2, 0)$) -- cycle;
            \end{pgfonlayer}
            \node[right = -2pt of 0, nlab]{$3$};
            \node[right = -2pt of 00, nlab]{$0$};
            \node[right = -2pt of 01, nlab]{$3$};
            \node[right = -2pt of 010, nlab]{$3$};
            \node[right = -2pt of 0100, nlab]{$1$};
            \node[right = -2pt of 0101, nlab]{$2$};
            \node[right = -2pt of 01010, nlab]{$1$};
            \node[right = -2pt of 01011, nlab]{$1$};
            \node[right = -2pt of 011, nlab]{$0$};
        \end{tikzpicture}
        \quad
        =
        \quad
        \left(
        \begin{tikzpicture}[
                    level distance=4ex,
                    sibling distance=6ex,
                    baseline=-4ex
                ]
                \tikzset{fit1/.style={rounded rectangle, inner sep=0.8pt, fill=yellow!40, opacity = .8},
                    nlab/.style={font= \scriptsize, color = red}}
                \node(0){$\circ$}
                child {node(00){$\const{I}$}}
                child {node(01){$\circ$}
                        child {node(010) {$a$}
                            }
                        child {node(011) {$\const{I}$}
                            }
                    }
                ;
                \begin{pgfonlayer}{background}
                    \node[fit1, rotate fit=-45, fit=(0)(01)] {};
                    \node[fit1, rotate fit=45, fit=(01)(010)] {};
                \end{pgfonlayer}
                \node[right = -2pt of 0, nlab]{$1$};
            \end{tikzpicture}\right)\left[\begin{tikzpicture}[
                    level distance=4ex,
                    sibling distance=6ex,
                    baseline=-4ex
                ]
                \tikzset{fit1/.style={rounded rectangle, inner sep=0.8pt, fill=yellow!40, opacity = .8},
                    fit2/.style={red!20, fill = red!20, thick},
                    fit3/.style={blue!20, fill = blue!20, thick},
                    nlab/.style={font= \scriptsize, color = red}}
                \node(010) {$\circ$}
                child {node(0100) {$a$}
                    }
                child {node(0101) {$\circ$}
                        child {node(01010) {$b$}
                            }
                        child {node(01011) {$c$}
                            }
                    }
                ;
                \begin{pgfonlayer}{background}
                    \draw [fit2] ($(0100.north) + (0,0)$) -- ($(0100.south west) + (-.2, 0)$) -- ($(0100.south east) + (.2, 0)$) -- cycle;
                    \draw [fit3] ($(0101.north) + (0,0)$) -- ($(01010.south west) + (-.2, 0)$) -- ($(01011.south east) + (.2, 0)$) -- cycle;
                \end{pgfonlayer}
                \node[right = -2pt of 0100, nlab]{$1$};
                \node[right = -2pt of 0101, nlab]{$2$};
            \end{tikzpicture} \quad / \quad a\right].\]
\end{example}
Using this decomposition iteratively, we have the following:
\begin{lemma}\label{lemma: 1 to k}
    Let $\termset \subseteq \allterm$ be a subterm-closed set.
    Assume that $\voset{\termset}{1}\quoset{\simalgclass{\algclass}}$ is finite.
    Then, for each $k \in \nat$, the set $\voset{\termset}{k}\quoset{\simalgclass{\algclass}}$ is finite.
\end{lemma}
\begin{proof}
    It suffices to prove:
    for all $k \ge 2$, $\voseteq{\termset}{k}\quoset{\simalgclass{\algclass}}$ is finite.
    By induction on $k$.
    We have:
    \begin{align*}
         & \card(\voseteq{\termset}{k}\quoset{\simalgclass{\algclass}})                             \\
         & \le \card(\set{\term_0\assign{f(\term_1, \dots, \term_n)}{a} \mid
            f^{(n)} \in \sig, \term_1, \dots, \term_n \in \voset{\termset}{k-1}, \term_0 \in \voset{\termset}{1}}\quoset{\simalgclass{\algclass}}) \tag{\Cref{lemma: decomposition}
        }                                                                                           \\
         & \le \sum_{f^{(n)} \in \sig}\card(\set{\term_0\assign{f(\term_1, \dots, \term_n)}{a} \mid
            \term_1, \dots, \term_n \in \voset{\termset}{k-1}, \term_0 \in \voset{\termset}{1}}\quoset{\simalgclass{\algclass}}).
    \end{align*}
    Then the set $\set{\term_0\assign{f(\term_1, \dots, \term_n)}{a} \mid
            \term_1, \dots, \term_n \in \voset{\termset}{k-1}, \term_0 \in \voset{\termset}{1}}\quoset{\simalgclass{\algclass}}$ is finite
    because $\voset{\termset}{k-1}\quoset{\simalgclass{\algclass}}$ is finite (by IH) and $\simalgclass{\algclass}$ satisfies the congruence law.
    Thus, the last term above is finite since $\sig$ is finite.
    Hence $\voseteq{\termset}{k}\quoset{\simalgclass{\algclass}}$ is finite.
\end{proof}

\newcommand{\bl}{\_}
\knowledgenewrobustcmd{\empword}{\cmdkl{\varepsilon}}
\knowledgenewrobustcmd{\context}[1]{\cmdkl{[}#1\cmdkl{]}}
\knowledgenewrobustcmd{\wlen}[1]{\cmdkl{\|}#1\cmdkl{\|}}

\subsection{The monoid of the $1$-variable-occurrence fragment}\label{section: monoid}
Thanks to \Cref{lemma: 1 to k}, we can focus on the \kl{$1$-variable-occurrence fragment}.
For the \kl{$1$-variable-occurrence fragment}, it suffices to consider a \kl{monoid}.
For a set $A$ of characters, we write $A^{*}$ for the set of all \intro{words} (i.e., finite sequences) over the language $A$.
We write $\word[1] \word[2]$ for the concatenation of \kl{words} $\word[1]$ and $\word[2]$
and write $\intro*\empword$ for the empty \kl{word}.
We write $\intro*\wlen{\word[1]}$ for the length of a \kl{word} $\word[1]$.
\begin{definition}\label{definition: unary maps}
    Let $\intro*\extpvsig$ be the (possibly infinite) set of characters defined by:
    \[\intro*\extpvsig \defeq \bigcup_{f^{(n)} \in \sig, i \in \range{1,n}}\set{f(\term_1, \dots, \term_{i-1}, \bl, \term_{i+1}, \dots, \term_n) \mid
            \forall j \in \range{1, n} \setminus \set{i},\  \term_j \in \voset{\allterm}{0}}.\]
    ($\bl$ denotes ``blank''.)
    For a word $\word \in \extpvsig^*$ and a term $\term \in \allterm$,
    let $\word[1]\intro*\context{\term}$ be the term defied by:
    \[\word[1]\intro*\context{\term} \defeq \begin{cases}
            f(\term_1, \dots, \term_{i-1}, \word[2]\context{\term}, \term_{i+1}, \dots, \term_n) & (\word[1] = f(\term_1, \dots, \term_{i-1}, \bl, \term_{i+1}, \dots, \term_n) \word[2]) \\
            \term                                                                                & (\word[1] = \empword)
        \end{cases}.\]
\end{definition}
\begin{example}[of \Cref{definition: unary maps}]
    If $\sig = \set{\circ^{(2)}, \const{I}^{(0)}}$, then $\extpvsig = \set{(\const{I} \circ \bl), ((\const{I} \circ \const{I}) \circ \bl), \dots (\bl \circ \const{I}), \dots}$.
    For example, if $\word[1] = ((\const{I} \circ \const{I}) \circ \bl) (\const{I} \circ \bl) (\bl \circ \const{I})$, we have:
    \begin{align*}
        \word[1]\context{a} = (((\const{I} \circ \const{I}) \circ \bl) (\const{I} \circ \bl) (\bl \circ \const{I})\context{a}) & =
        (\const{I} \circ \const{I}) \circ ((\const{I} \circ \bl) (\bl \circ \const{I})\context{a})                                 \\
                                                                                                                               & =
        (\const{I} \circ \const{I}) \circ (\const{I} \circ ((\bl \circ \const{I})\context{a})) = (\const{I} \circ \const{I}) \circ (\const{I} \circ (a \circ \const{I})).
    \end{align*}
\end{example}
\begin{proposition}\label{proposition: decomposition by monoid}
    For all $\term \in \voset{\allterm}{1}$, there are $\word \in \extpvsig^*$ and $\term[2] \in \sig^{(0)} \cup \vsig$ s.t.\ $\term = \word[1]\context{\term[2]}$.
\end{proposition}
\begin{proof}
    By easy induction on $\term$.
\end{proof}
\begin{definition}\label{definition: monoid equation}
    Let $\intro*\extpsim{\algclass}$ be the equivalence relation on $\extpvsig^*$ defined by:
    \[\word[1] \intro*\extpsim{\algclass} \word[2] \defiff \word[1]\context{a} \simalgclass{\algclass} \word[2]\context{a} \mbox{ where $a \in \vsig$ is any variable.}\]
\end{definition}
\begin{lemma}\label{lemma: monoid}
    If $\extpvsig^*\quoset{\extpsim{\algclass}}$ is finite, then $\voset{\allterm}{1}\quoset{\simalgclass{\algclass}}$ is finite.
\end{lemma}
\begin{proof}
    By \Cref{proposition: decomposition by monoid} (and that the set $\sig^{(0)} \cup \vsig$ is finite).
\end{proof}
Moreover, if $\voset{\allterm}{0}\quoset{{\simalgclass{\algclass}}}$ is finite, it suffices to consider a finite subset of $\extpvsig$, as follows:
\begin{lemma}\label{lemma: finite generator}
    Assume that $\voset{\allterm}{0}\quoset{\simalgclass{\algclass}}$ is finite.
    Let $\termset_0 = \set{\term_1, \dots, \term_n} \subseteq \voset{\allterm}{0}$ be such that $\voset{\allterm}{0}\quoset{\simalgclass{\algclass}} = \set{\quoclass{\term_1}_{\simalgclass{\algclass}}, \dots, \quoclass{\term_n}_{\simalgclass{\algclass}}}$.
    Let $\intro*\extpvsigzero \subseteq \extpvsig$ be the finite set defined by:
    \[\reintro*\extpvsigzero \defeq \bigcup_{f^{(n)} \in \sig, i \in \range{1,n}} \set{f(\term_1, \dots, \term_{i-1}, \bl, \term_{i+1}, \dots, \term_n) \mid
            \forall j \in \range{1, n} \setminus \set{i},\  \term_j \in \termset_0}.\]
    Then $\extpvsigzero^*\quoset{\extpsim{\algclass}}$ is finite $\Longrightarrow$ $\extpvsig^*\quoset{\extpsim{\algclass}}$ is finite.
\end{lemma}
\begin{proof}
    For every $a \in \extpvsig$, there is $b \in \extpvsigzero$ s.t.\  $a \extpsim{\algclass} b$.
    By the congruence law of $\extpsim{\algclass}$, for every $\word \in \extpvsig^*$, there is some $\word[2] \in \extpvsigzero^*$ s.t.\  $\word[2] \extpsim{\algclass} \word$.
    Since $\extpvsig^*\quoset{\extpsim{\algclass}}$ is finite, this completes the proof.
\end{proof}

\begin{example}[of \Cref{lemma: finite generator}]
    If $\sig = \set{\circ^{(2)}, \const{I}^{(0)}}$ (so, $\extpvsig = \set{(\const{I} \circ \bl), ((\const{I} \circ \const{I}) \circ \bl), \dots (\bl \circ \const{I}), \dots}$) and $\algclass$ is the class of all \kl{monoids},
    we have: $\voset{\allterm}{0}\quoset{\simalgclass{\algclass}} = \set{\quoclass{\const{I}}_{\simalgclass{\algclass}}}$.
    Thus the set $\extpvsigzero = \set{(\const{I} \circ \bl), (\bl \circ \const{I})}$ is sufficient for considering the finiteness of $\extpvsig^*\quoset{\extpsim{\algclass}}$.
\end{example}
Thus, to prove the finiteness of $\voset{\allterm}{k}\quoset{\simalgclass{\algclass}}$,
it suffices to prove that both $\voset{\allterm}{0}\quoset{\simalgclass{\algclass}}$ and $\extpvsigzero^*\quoset{\extpsim{\algclass}}$ are finite:
\begin{lemma}\label{lemma: decidable}
    If $\voset{\allterm}{0}\quoset{\simalgclass{\algclass}}$ and $\extpvsigzero^*\quoset{\extpsim{\algclass}}$ are finite,
    then for each $k \in \nat$,
    the set $\voset{\allterm}{k}\quoset{\simalgclass{\algclass}}$ is finite (hence, the equational theory of $\voset{\allterm}{k}$ over $\algclass$ is decidable).
\end{lemma}
\begin{proof}
    We have:
    $\voset{\allterm}{0}\quoset{\simalgclass{\algclass}}$ and $\extpvsigzero^*\quoset{\extpsim{\algclass}}$ are finite
    $\Longrightarrow$
    $\extpvsig^*\quoset{\extpsim{\algclass}}$ is finite (by \Cref{lemma: finite generator})
    $\Longrightarrow$
    $\voset{\allterm}{1}\quoset{\simalgclass{\algclass}}$ is finite (by \Cref{lemma: monoid})
    $\Longrightarrow$
    $\voset{\allterm}{k}\quoset{\simalgclass{\algclass}}$ is finite (by \Cref{lemma: 1 to k}).
    Hence by \Cref{proposition: finite to decidable}.
\end{proof}

\subsection{Finiteness from finding equations}
For the finiteness of $\extpvsigzero^*\quoset{\extpsim{\algclass}}$,
we consider finding equations $\set{\tuple{\word[1]_i, \word[2]_i} \mid i \in I}$ and then applying the following:
\begin{lemma}\label{lemma: cofiniteness}
    Let $\extpvsigzero \subseteq \extpvsig$ be a finite set.
    Let $(<) \subseteq (\extpvsigzero^*)^2$ be a well-founded relation s.t.\
    \begin{itemize}
        \item $<$ satisfies the congruence law (i.e., $\word[2] < \word[2]' \Longrightarrow \word[1]\word[2]\word[1]' < \word[1]\word[2]'\word[1]'$);
        \item $<$ has no infinite antichains.\footnote{This assumption is used only in the direction of \ref{lemma: cofiniteness 2}$\Longrightarrow$\ref{lemma: cofiniteness 1}.}
    \end{itemize}
    Then, the following are equivalent:
    \begin{enumerate}
        \item \label{lemma: cofiniteness 1} There is a finite set $\set{\tuple{\word[1]_i, \word[2]_i} \mid i \in I} \subseteq (<) \cap (\extpsim{\algclass})$ such that
              the language $\extpvsigzero^* (\bigcup_{i \in I} \word[2]_i) \extpvsigzero^*$ over the alphabet $\extpvsigzero$ is \kl{cofinite}.\footnote{A language $L$ over an  alphabet $A$ is \intro{cofinite} if its complemented language $A^* \setminus L$ is finite.}
        \item \label{lemma: cofiniteness 2} $\extpvsigzero^*\quoset{\extpsim{\algclass}}$ is finite.
    \end{enumerate}
\end{lemma}
\newcommand{\subw}{\mathop{\mathrm{Subw}}}
\begin{proof}
    \ref{lemma: cofiniteness 1}$\Longrightarrow$\ref{lemma: cofiniteness 2}:
    By induction on the well-founded relation $<$, we prove:
    For every $\word[1] \in \extpvsigzero^*$, there is some $\word[2] \in \extpvsigzero^* \setminus (\extpvsigzero^* (\bigcup_{i \in I} \word[2]_i) \extpvsigzero^*)$ such that
    $\word[1] \extpsim{\algclass} \word[2]$.
    If $\word \in \extpvsigzero^* \setminus (\extpvsigzero^* (\bigcup_{i \in I} \word[2]_i) \extpvsigzero^*)$, by letting $\word[2] = \word$.
    Otherwise, since $\word \in \extpvsigzero^* (\bigcup_{i \in I} \word[2]_i) \extpvsigzero^*$,
    there are $i \in I$ and $\word', \word'' \in \extpvsigzero^*$ such that $\word = \word' \word[2]_i \word''$.
    By $\word' \word[1]_i \word'' < \word' \word[2]_i \word''$ (the congruence law of $<$) and IH,
    there is $\word[2] \in \extpvsigzero^* \setminus (\extpvsigzero^* (\bigcup_{i \in I} \word[2]_i) \extpvsigzero^*)$ s.t.\  $\word' \word[1]_i \word''  \extpsim{\algclass} \word[2]$.
    We also have $\word' \word[2]_i \word'' \extpsim{\algclass} \word' \word[1]_i \word''$ (by $\word[2]_i \extpsim{\algclass} \word[1]_i$ with the congruence law of $\extpsim{\algclass}$).
    Thus $\word' \word[2]_i \word'' \extpsim{\algclass} \word[2]$ (by transitivity of $\extpsim{\algclass}$).

    \ref{lemma: cofiniteness 2}$\Longrightarrow$\ref{lemma: cofiniteness 1}:
    Let $W \defeq \bigcup_{X \in \extpvsigzero^*\quoset{\extpsim{\algclass}}}\set{\word \in X \mid \mbox{$\word$ is minimal w.r.t.\ $(<) \cap X^2$}}$.
    Let $\subw(W)$ be the subword closure of $W$ (i.e., the minimal set $W' \supseteq W$ s.t.\ $w'ww'' \in W' \Longrightarrow w \in W'$).
    Let $V \defeq (\subw(W) \extpvsigzero) \setminus \subw(W)$.
    Then, $(\extpvsigzero^* V \extpvsigzero^*) = \extpvsigzero^* \setminus \subw(W)$ holds, as follows.
    For $\subseteq$:
    Let $w \in \extpvsigzero^*$, $v \in V$, $w' \in \extpvsigzero^*$.
    If we assume $w v w' \in \subw(W)$, then $v \in \subw(W)$, but this contradicts $v \in V$; thus, $w v w' \not\in \subw(W)$.
    For $\supseteq$:
    By induction on the length of $w \in \extpvsigzero^* \setminus \subw(W)$.
    If $w \in V$, clear.
    Otherwise (i.e., $w \not\in V$),
    let $w = w' a$ (note that $w \neq \empword$, because $\empword \in \subw(W)$ always by that $W$ is not empty).
    Then, $w' \in \extpvsigzero^* \setminus \subw(W)$
    (if not, since $w' \in \subw(W)$ and $w \not\in \subw(W)$, $w \in V$, thus reaching a contradiction).
    Thus by IH, $w' \in \extpvsigzero^* V \extpvsigzero^*$, and thus $w \in \extpvsigzero^* V \extpvsigzero^*$.
    Hence, we have $\extpvsigzero^* \setminus (\extpvsigzero^* V \extpvsigzero^*) = \subw(W)$.
    Now, the set $W$ is finite because $\extpvsigzero^*\quoset{\extpsim{\algclass}}$ is finite and for each $X \in \extpvsigzero^*\quoset{\extpsim{\algclass}}$, the number of minimal elements $\word$ is finite (because $<$ has no infinite antichains);
    thus $\subw(W)$ is finite; thus $V$ is finite.
    Let $V = \set{v_1, \dots, v_n}$.
    For every $i \in \range{1, n}$, there is $w_i \in W \cap \quoclass{v_i}_{\extpsim{\algclass}}$ s.t.\ $w_i < v_i$, because
    $v_i$ is not minimal w.r.t.\ $(<) \cap \quoclass{v_i}_{\extpsim{\algclass}}^2$.
    Thus, $\set{\tuple{w_i, v_i} \mid i \in \range{1, n}}$ is the desired set.
\end{proof}
The \emph{shortlex order} (aka \emph{length-lexicographical order}) is an example of $<$ in \Cref{lemma: cofiniteness} (because it is a well-ordering \cite[Def.\ 2.2.3]{bookStringRewritingSystems1993}
and its congruence raw is also easy).
While it is undecidable whether a given (finitely presented) monoid $\extpvsigzero^*\quoset{\extpsim{\algclass}}$ is finite \cite{markovImpossibilityCertainAlgorithms1947} (see also \cite[Thm.\ 7.3.7 with Def.\ 7.3.2(b)]{bookStringRewritingSystems1993}) in general (cf.\ \ref{lemma: cofiniteness 2} of \Cref{lemma: cofiniteness}),
it is decidable (in linear time) whether the language of a given regular expression of the form $\extpvsigzero^* (\bigcup_{i \in I} \word[2]_i) \extpvsigzero^*$ is \kl{cofinite} (cf.\ \ref{lemma: cofiniteness 1} of \Cref{lemma: cofiniteness}):
\begin{proposition}\label{proposition: linear time}
    The following is decidable in linear time (more precisely, $\mathcal{O}(n)$ time on a RAM machine for $n$ the number of symbols in the given regular expression):
    Given a regular expression of the form $\extpvsigzero^* (\bigcup_{i \in I} \word[2]_i) \extpvsigzero^*$ over the alphabet $\extpvsigzero$, is its language cofinite?
\end{proposition}
\begin{proof}[Proof Sketch]
    By the Aho-Corasick algorithm,
    we can construct a deterministic finite automaton (DFA) from a given regular expression of the form $\extpvsigzero^* (\bigcup_{i \in I} \word[2]_i) \extpvsigzero^*$ in linear time \cite[Sect.\ 8]{ahoEfficientStringMatching1975}.
    By taking the complemented language of the DFA,
    it suffices to show that the following problem is in $\mathcal{O}(n)$ time: given a DFA with $n$ states, is its language finite?
    Then we can give the following algorithm:
    From the graph induced by the DFA, remove all the states not reachable from the starting state and remove all the states not reachable to any accepting states by using the depth-first search;
    check whether there exists some cycle in the graph by the depth-first search.
\end{proof}
Thus, thanks to \ref{lemma: cofiniteness 1}$\Longrightarrow$\ref{lemma: cofiniteness 2} of \Cref{lemma: cofiniteness}, we can focus on finding a finite set of equations.
While it is undecidable in general whether there exists such a set,
we give a possibly non-terminating pseudo-code in \Cref{algorithm: finiteness proof},
which can help to find equations (e.g., \Cref{figure: equations}).\footnote{Usually, to calculate $\extpsim{\algclass}$ is a bottleneck. For relaxing this problem, for example, hashing words by using some algebras in $\algclass$ is practically useful for reducing the number of $\extpsim{\algclass}$ calls (since if the hash of two words $\word[1], \word[2]$ are different, then we immediately have that $\word[1], \word[2]$ are not equivalent w.r.t.\  $\extpsim{\algclass}$).}

\alglanguage{pseudocode}
\algrenewcommand{\algorithmicindent}{1em}
\algnewcommand{\algorithmicgoto}{\textbf{goto}}%
\algnewcommand{\Goto}[1]{\algorithmicgoto~\ref{#1}}%
\begin{algorithm}[H]
    \caption{Possibly non-terminating pseudo-code for ensuring the finiteness of $\extpvsigzero^*\quoset{\extpsim{\algclass}}$.}
    \label{algorithm: finiteness proof}
    \begin{algorithmic}[1]
        \Require Given $\sig, \vsig$, $\extpsim{\algclass}$, $\extpvsigzero$. \Comment{$\extpvsigzero$ is a finite alphabet of \Cref{lemma: finite generator}.}
        \Ensure Is $\extpvsigzero^*\quoset{\extpsim{\algclass}}$ finite?
        \State $\Gamma$ $\gets$ $\emptyset$ \quad \Comment{We use $I, \word[1]_i, \word[2]_i$ to denote $\Gamma = \set{\tuple{\word[1]_i, \word[2]_i} \mid i \in I}$.}
        \While{the language $\extpvsigzero^* (\bigcup_{i \in I} \word[2]_i) \extpvsigzero^*$ over $\extpvsigzero$ is not \kl{cofinite}}
        \State $\tuple{\word[1], \word[2]}$ $\gets$ a fresh pair in $X^* \times (X^* \setminus (X^* (\bigcup_{i \in I} \word[2]_i) X^*))$
        s.t.\ $\word[1] < \word[2]$ \Comment{$<$ is a binary relation in (\Cref{lemma: cofiniteness}).}
        \If{$\word[1] \extpsim{\algclass} \word[2]$} $\Gamma$ $\gets$ $\Gamma \cup \set{\tuple{\word[1], \word[2]}}$
        \EndIf
        \EndWhile
        \State \Return $\const{True}$
    \end{algorithmic}
\end{algorithm}
\begin{remark}\label{remark: algorithm: finiteness proof}
    When $\extpsim{\algclass}$ is given as a total recursive function, \Cref{algorithm: finiteness proof} is a semi-algorithm (that is, if $\extpvsigzero^*\quoset{\extpsim{\algclass}}$ is finite, the algorithm is terminated and returns $\const{True}$; otherwise, not terminated).
    This is because \ref{lemma: cofiniteness 2}$\Longrightarrow$\ref{lemma: cofiniteness 1} of \Cref{lemma: cofiniteness} also holds.
\end{remark}

\section{The calculus of relations with bounded dot-dagger alternation}\label{section: CoR}
In the remaining part of this paper, as a case study of the \kl{$k$-variable-occurrence fragment} presented in \Cref{section: 1-variable}, we consider the calculus of relations with bounded dot-dagger alternation.
In this section, we recall the definitions of the calculus of relations (\kl{CoR}) and the \kl{dot-dagger alternation hierarchy}.

\subsection{CoR: syntax and semantics}
We fix $\vsig$ as a non-empty \emph{finite} set of variables.
Consider the finite algebraic signature $\intro*\CoRsig \defeq \set{\bot^{(0)}, \top^{(0)}, -^{(1)}, \cup^{(2)}, \cap^{(2)}, \const{I}^{(0)}, \const{D}^{(0)}, \cdot^{(2)}, \dagger^{(2)}}
    \cup \set{\pi^{(1)} \mid \mbox{$\pi$ is a map in $\range{1, 2}^{\range{1, 2}}$}}$ (we consider algebras of binary relations and each $\pi$ is used for a projection of binary relations).
The set $\intro*\CoRallterm$ of \intro{CoR terms} is defined as follows:
\begin{align*}
    \intro*\CoRallterm \ni \term[1], \term[2], \term[3] \Coloneqq a & \mid \bot \mid \top \mid \term[2] \cup \term[3] \mid \term[2] \cap \term[3] \mid \term[2]^{-}                                                                                \\
                                                                    & \mid \const{I} \mid \const{D} \mid \term[2] \cdot \term[3] \mid \term[2] \dagger \term[3] \mid \term[2]^{\pi} \tag*{($a \in \vsig$, $\pi \in \range{1, 2}^{\range{1, 2}}$).}
\end{align*}
Additionally, for a term $\term$, we use $\term^{\smile}$ to denote the term $\term^{\smile} \defeq \term^{\set{1 \mapsto 2, 2 \mapsto 1}}$.
Here, we use the infix notation for binary operators, the superscript notation for unary operators, and parenthesis in ambiguous situations, as usual.

For binary relations $R, S$ on a set $W$,
the \emph{identity relation} $\const{I}_{W}$ on $W$, the \emph{difference relation} $\const{D}_{W}$ on $W$,
the \emph{(relational) composition} (\emph{relative product}) $R \cdot S$, the \emph{dagger} (\emph{relative sum}) $R \dagger S$, and the \emph{projection} $R^{\pi}$ are defined by:
\begin{align*}
    \const{I}_{W} & \defeq \set{\tuple{x,y} \in W^2 \mid x = y}                         \tag{identity}                                               \\
    \const{D}_{W} & \defeq \set{\tuple{x,y} \in W^2 \mid x \neq y}                      \tag{difference}                                             \\
    R \cdot S     & \defeq \set{\tuple{x, y} \in W^2 \mid \exists z \in W,\  \tuple{x, z} \in R \ \land \ \tuple{z, y} \in S} \tag{relative product} \\
    R \dagger S   & \defeq \set{\tuple{x, y}  \in W^2\mid \forall z \in W,\  \tuple{x, z} \in R \ \lor \ \tuple{z, y} \in S} \tag{relative sum}      \\
    R^{\pi}       & \defeq \set{ \tuple{x_{1}, x_{2}}  \in W^2 \mid \tuple{x_{\pi(1)}, x_{\pi(2)}} \in R} \tag*{(projection).}
\end{align*}
A \intro{structure} $\model$ is a tuple $\tuple{|\model|, \set{a^{\model}}_{a \in \vsig}}$, where
$|\model|$ is a non-empty set and $a^{\model} \subseteq |\model|^2$ is a binary relation for each $a \in \vsig$.
For a \kl{structure} $\model$,
the \emph{binary relation} map $\intro*\jump{\bl}_{\model} \colon \CoRallterm \to \powerset(|\model|^2)$ is the unique homomorphism extending $\jump{a}_{\model} = a^{\model}$ w.r.t.\ the set-theoretic operators and the aforementioned binary relation operators; i.e.,
$\intro*\jump{\term}_{\model}$ is defined as follows:
\begin{align*}
    \jump{a}_{\model}                         & \defeq  a^{\model} \hspace{.5em}(a \in \vsig)  \hspace{1.5em}      \jump{\bot}_{\model}       \defeq \emptyset\hspace{1.5em}
    \jump{\top}_{\model}                      \defeq |\model|^2     \hspace{1.5em}       \jump{\const{I}}_{\model}                                                         \defeq \const{I}_{|\model|}
    \qquad \jump{\const{D}}_{\model}                                                                                                                                                            \defeq \const{D}_{|\model|} \span\span\span\span \\
    \jump{\term[1] \cup \term[2]}_{\model}    & \defeq \jump{\term[1]}_{\model} \cup \jump{\term[2]}_{\model}                                                                &
    \jump{\term[1] \cap \term[2]}_{\model}    & \defeq \jump{\term[1]}_{\model} \cap \jump{\term[2]}_{\model}                                                                &
    \jump{\term[1]^{-}}_{\model}              & \defeq |\model|^2 \setminus \jump{\term}_{\model}                                                                                                                                                \\
    \jump{\term[1] \cdot \term[2]}_{\model}   & \defeq \jump{\term[1]}_{\model} \cdot \jump{\term[2]}_{\model}                                                               &
    \jump{\term[1] \dagger \term[2]}_{\model} & \defeq \jump{\term[1]}_{\model} \dagger \jump{\term[2]}_{\model}                                                             &
    \jump{\term[1]^{\pi}}_{\model}            & \defeq \jump{\term[1]}_{\model}^{\pi}.
\end{align*}
It is well-known that w.r.t.\ binary relations, \kl{CoR} has the same expressive power as the three-variable fragment of first-order logic with equality:\footnote{Namely, for every formula $\fml$ with two distinct free variables $z_1, z_2$ in the three-variable fragment of first-order logic with equality, there is $\term \in \CoRallterm$ such that for all $\model$, $\jump{\lambda z_1 z_2. \fml}_{\model} = \jump{\term}_{\model}$.
    Conversely, for every $\term \in \CoRallterm$, there is $\fml$ such that for all $\model$, $\jump{\term}_{\model} = \jump{\lambda z_1 z_2. \fml}_{\model}$.
    Here, $\jump{\lambda z_1 z_2. \fml}_{\model} \defeq \set{\tuple{x, y} \in |\model|^2 \mid \mbox{$\fml$ is true on $\model$ if $z_1, z_2$ are mapped to $x, y$, respectively}
        }$.}
\begin{proposition}[\cite{Tarski1941,Tarski1987,Givant07}]\label{proposition: CoR and FO3}
    W.r.t.\ binary relations, the expressive power of $\CoRallterm$ is equivalent to that of the three-variable fragment of first-order logic with equality.
\end{proposition}

\knowledgenewrobustcmd\simalgclassx[1]{\mathbin{\cmdkl{\sim}_{#1}}}
\newrobustcmd{\notsimalgclassx}[1]{\mathbin{{\kl[\simalgclassx]{\not\sim}}_{#1}}}
\knowledge{\simalgclassx}{automatic in command}

Let $\intro*\REL$ be the class of all \kl{structures}.
Let $\reintro*\REL_{\ge m}$ (resp.\ $\reintro*\REL_{\le m}$) be the class of \kl{structures} $\model$ of $\card|\model| \ge m$ (resp.\ $\card|\model| \le m$).
For $\algclass \subseteq \REL$, the equivalence relation $\intro*\simalgclassx{\algclass}$ on $\CoRallterm$ is defined by:
$\term[1] \reintro*\simalgclassx{\algclass} \term[2] \defiff \jump{\term[1]}_{\model} = \jump{\term[2]}_{\model}$ for every $\model \in \algclass$.
For $\termset \subseteq \CoRallterm$, the \reintro[equational theory]{equational theory of $\termset$ over $\algclass$} is the set $\set{\tuple{\term[1], \term[2]} \in \termset^2 \mid \term[1] \simalgclassx{\algclass} \term[2]}$.
We mainly consider $\simalgclassx{\REL}$: the \kl{equational theory} over $\REL$.
The following are some instances w.r.t.\ $\simalgclassx{\REL}$:
\begin{align*}
    a \cdot (b \cdot c) & \simalgclassx{\REL} (a \cdot b) \cdot c   & a \cap (b \cup c)     & \simalgclassx{\REL} (a \cap b) \cup (a \cap c) & (a^{\smile})^{\smile}              & \simalgclassx{\REL} a                              \\
    a \cdot a^{\smile}  & \notsimalgclassx{\REL} a^{\smile} \cdot a & a \cdot (b \dagger c) & \notsimalgclassx{\REL} (a \cdot b) \dagger c   & a^{\set{1 \mapsto 1, 2 \mapsto 1}} & \simalgclassx{\REL} (a \cap \const{I}) \cdot \top.
\end{align*}

The following propositions hold because for each $m \in \nat$, the number of \kl{structures} $\model$ of $\card|\model| \le m$ is finite up to isomorphism and each \kl{structure} is finite.
\begin{proposition}\label{proposition: le m}
    For each $m \in \nat$, $\CoRallterm\quoset{\simalgclassx{\REL_{\le m}}}$ is finite.
\end{proposition}
\begin{proposition}\label{proposition: eliminate fin}
    Let $\termset \subseteq \CoRallterm$ be a subterm-closed set and $m \ge 1$.
    Then, $\termset\quoset{\simalgclassx{\REL}}$ is finite $\iff$ $\termset\quoset{\simalgclassx{\REL_{\ge m}}}$ is finite.
    Additionally, the \kl{equational theory} of $\termset$ over $\REL$ is decidable $\iff$ the \kl{equational theory} of $\termset$ over $\REL_{\ge m}$ is decidable.
\end{proposition}
\begin{proof}
    Because $\term[1] \simalgclassx{\REL} \term[2] \iff \term[1] \simalgclassx{\REL_{\le m - 1}} \term[2] \land \term[1] \simalgclassx{\REL_{\ge m}} \term[2]$.
    By \Cref{proposition: le m} with \Cref{proposition: finite to decidable}, $\termset\quoset{\simalgclassx{\REL_{\le m - 1}}}$ is finite and
    the \kl{equational theory} of $\termset$ over $\REL_{\le m - 1}$ is decidable.
\end{proof}

\subsection{The dot-dagger alternation hierarchy}
\begin{definition}[the \intro*\kl{dot-dagger alternation hierarchy} {\cite{nakamuraExpressivePowerSuccinctness2022}}]\label{definition: dot-dagger}
    The sets, $\set{\intro*\CoRSigma_{n}, \intro*\CoRPi_{n}}_{n \in \nat}$, are the minimal sets satisfying the following:
    \begin{itemize}
        \item $\CoRSigma_0 = \CoRPi_0 = \set{\term \in \CoRallterm \mid \mbox{$\term$ does not contain $\cdot$ nor $\dagger$}}$;
        \item For $n \ge 0$, $\CoRSigma_n \cup \CoRPi_n \subseteq \CoRSigma_{n+1} \cap \CoRPi_{n+1}$;
        \item For $n \ge 1$, if $\term[2], \term[3] \in \CoRSigma_n$, then $\term[2] \cup \term[3], \term[2] \cap \term[3], \term[2] \cdot \term[3], \term[2]^{\pi} \in \CoRSigma_n$ and $\term[2] \dagger \term[3] \in \CoRPi_{n+1}$;
        \item For $n \ge 1$, if $\term[2], \term[3] \in \CoRPi_n$, then $\term[2] \cup \term[3], \term[2] \cap \term[3], \term[2] \dagger \term[3], \term[2]^{\pi} \in \CoRPi_n$ and $\term[2] \cdot \term[3] \in \CoRSigma_{n+1}$.
    \end{itemize}
\end{definition}
For example, $a \cdot b \in \CoRSigma_{1}$ and $a \cdot (b \dagger c) \in \CoRSigma_{2}$ (the term $a \cdot b$ means that for some $z$, $a(x, z)$ and $b(z, y)$.
The term $a \cdot (b \dagger c)$ means that for some $z$, for every $w$, $a(x, z)$ and ($b(z, w)$ or $c(w, y))$.
Here, $x$ and $y$ indicate the source and the target, respectively, and each $a'(x', y')$ denotes that there is an $a'$-labelled edge from $x'$ to $y'$).
The \kl{dot-dagger alternation hierarchy} is an analogy of the quantifier alternation hierarchy in first-order logic (by viewing $\cdot$ as $\exists$ and $\dagger$ as $\forall$).
This provides a fine-grained analogy of \Cref{proposition: CoR and FO3} w.r.t.\ the number of quantifier alternations, as follows:
\begin{proposition}[{\cite[Cor.\ 3.14]{nakamuraExpressivePowerSuccinctness2022}}; cf.\ \Cref{proposition: CoR and FO3}]\label{proposition: dot-dagger}
    W.r.t.\ binary relations, the expressive power of $\CoRSigma_{n}$ (resp.\ $\CoRPi_{n}$) is equivalent to that of the level $\Sigma_n$ (resp.\ $\Pi_n$) in the quantifier alternation hierarchy of the three-variable fragment of first-order logic with equality.
\end{proposition}
Because there are recursive translations for \Cref{proposition: dot-dagger} \cite{nakamuraExpressivePowerSuccinctness2022}, the following (un-)decidability results follow from those in first-order logic.
\begin{proposition}\label{proposition: decidable/undecidable}
    The equational theory of $\CoRSigma_{n}$ (resp.\ $\CoRPi_{n}$) is decidable if $n \le 1$
    and is undecidable if $n \ge 2$.
\end{proposition}
\begin{proof}[Proof Sketch]
    When $\vsig$ is a countably infinite set, they follow from the \kl{BSR class} $[\exists^*\forall^*, \mathrm{all}]_{=}$ \cite{Bernays1928, Ramsey1930} and the reduction class {$[\forall \exists \land \forall^3, (\omega, 1)]$} {\cite[Cor.\ 3.1.19]{Borger1997}}.
    We can strengthen this result even if $\card \vsig = 1$ by using a variant of the translation in \cite[Lem.\ 11]{Nakamura2019} for encoding countably infinitely many variables by one variable.
    (See \ifiscameraready \cite{nakamuraFiniteVariableOccurrence2023arxiv} \else \Cref{section: proposition: decidable/undecidable} \fi for more details.)
\end{proof}

\section{On the $k$-variable-occurence fragment of $\CoRSigma_{n}$}\label{section: CoR k}
We now consider $\voset{(\CoRSigma_{n})}{k}$: the $k$-variable-occurrence fragment of the level $\CoRSigma_{n}$ in the dot-dagger alternation hierarchy.
Clearly, $\CoRSigma_{n} = \bigcup_{k \in \nat} \voset{(\CoRSigma_{n})}{k}$.
While the equational theory of $\CoRSigma_{n}$ is undecidable in general (\Cref{proposition: decidable/undecidable}), we show that the equational theory of $\voset{(\CoRSigma_{n})}{k}$ is decidable (\Cref{corollary: main}).
Our goal in this section is to show the following:
\begin{theorem}\label{theorem: main}
    For each $n, k \in \nat$, $\voset{(\CoRSigma_{n})}{k}\quoset{\simalgclassx{\REL}}$ is finite.
\end{theorem}
Combining with \Cref{proposition: finite to decidable} yields the following decidability and complexity upper bound.
The complexity lower bound is because the equational theory can encode the \emph{boolean sentence value problem} \cite{bussBooleanFormulaValue1987} (even if $n = k = 0$), as
a given boolean sentence $\fml$ is true iff $\term \sim_{\REL} \top$, where $\term$ is the term obtained from $\fml$ by replacing $\land, \lor, \const{T}, \const{F}$ with $\cap, \cup, \top, \bot$, respectively.
\begin{corollary}\label{corollary: main}
    For $n, k \in \nat$, the \kl{equational theory} of $\voset{(\CoRSigma_{n})}{k}$ over $\REL$ is decidable.
    Moreover, it is complete for DLOGTIME-uniform $\mathrm{NC}^1$ under DLOGTIME reductions if the input is given as a well-bracketed string.
\end{corollary}
To prove \Cref{theorem: main}, we consider the finiteness of $\voset{\CoRallterm}{0}\quoset{\simalgclassx{\REL}}$ in \Cref{section: CoR 0}
and the finiteness of a monoid for $\voset{(\CoRSigma_{n})}{k}\quoset{\simalgclassx{\REL}}$ in \Cref{section: CoR monoid}, respectively (cf.\ \Cref{lemma: decidable}).

\subsection{On the finiteness of $\voset{\CoRallterm}{0}$} \label{section: CoR 0}
For the finiteness of $\voset{\CoRallterm}{0}\quoset{\simalgclassx{\REL}}$, by \Cref{proposition: eliminate fin}, it suffices to show the following:
\begin{lemma}\label{lemma: CoR 0}
    $\voset{\CoRallterm}{0}\quoset{\simalgclassx{\REL_{\ge 3}}} = \set{[\bot]_{\simalgclassx{\REL_{\ge 3}}}, [\top]_{\simalgclassx{\REL_{\ge 3}}}, [\const{I}]_{\simalgclassx{\REL_{\ge 3}}}, [\const{D}]_{\simalgclassx{\REL_{\ge 3}}}}$.
\end{lemma}
\begin{proof}
    W.r.t.\ $\simalgclassx{\REL_{\ge 3}}$,
    we prove that the four elements are closed under each operator.
    For the operators $\cap, -, \cdot, \smile$, this is shown by the following Cayley tables:
    \begin{center}
        \hfill
        \begin{tabular}{|c||c|c|c|c|}
            \hline $\cap$       & $\top$      & $\bot$ & $\const{I}$ & $\const{D}$ \\
            \hline
            \hline$\top$        & $\top$      & $\bot$ & $\const{I}$ & $\const{D}$ \\
            \hline $\bot$       & $\bot$      & $\bot$ & $\bot$      & $\bot$      \\
            \hline  $\const{I}$ & $\const{I}$ & $\bot$ & $\const{I}$ & $\bot$      \\
            \hline $\const{D}$  & $\const{D}$ & $\bot$ & $\bot$      & $\const{D}$ \\
            \hline
        \end{tabular}
        \hfill
        \begin{tabular}{|c||c|c|c|c|}
            \hline $-$         & \diagbox{}{} \\
            \hline
            \hline $\top$      & $\bot$       \\
            \hline $\bot$      & $\top$       \\
            \hline $\const{I}$ & $\const{D}$  \\
            \hline $\const{D}$ & $\const{I}$  \\
            \hline
        \end{tabular}
        \hfill
        \begin{tabular}{|c||c|c|c|c|}
            \hline $\cdot$      & $\top$ & $\bot$ & $\const{I}$ & $\const{D}$ \\
            \hline
            \hline $\top$       & $\top$ & $\bot$ & $\top$      & $\top$      \\
            \hline $\bot$       & $\bot$ & $\bot$ & $\bot$      & $\bot$      \\
            \hline  $\const{I}$ & $\top$ & $\bot$ & $\const{I}$ & $\const{D}$ \\
            \hline $\const{D}$  & $\top$ & $\bot$ & $\const{D}$ & $\top$      \\
            \hline
        \end{tabular}
        \hfill
        \begin{tabular}{|c||c|c|c|c|}
            \hline $\smile$     & \diagbox{}{} \\
            \hline
            \hline  $\top$      & $\top$       \\
            \hline  $\bot$      & $\bot$       \\
            \hline  $\const{I}$ & $\const{I}$  \\
            \hline  $\const{D}$ & $\const{D}$  \\
            \hline
        \end{tabular}
        \hfill
        \phantom{.}
    \end{center}
    Note that $\const{D} \cdot \const{D} \simalgclassx{\REL_{\ge 3}} \top$ holds thanks to ``$\ge 3$''.
    When $\card(|\model|) \ge 3$, we have: $\tuple{x, y} \in \jump{\const{D} \cdot \const{D}}_{\model}$
    iff ($\exists z \in |\model|, z \neq x \land z \neq y$) iff $|\model| \setminus \set{x, y} \neq \emptyset$ iff $\const{True}$ (cf.\ \Cref{remark: elim finite}).
    (Similarly for $\const{D} \cdot \top \simalgclassx{\REL_{\ge 2}} \top$.)
    For the other operators ($\cup, \dagger, \pi$), they can be expressed by using $\cap, -, \cdot, \smile$ as follows:
    $\term[1] \cup \term[2] \simalgclassx{\REL} (\term[1]^{-} \cap \term[2]^{-})^{-}$,
    $\term[1] \dagger \term[2] \simalgclassx{\REL} (\term[1]^{-} \cdot \term[2]^{-})^{-}$,
    $\term^{\set{1 \mapsto 1, 2 \mapsto 2}} \simalgclassx{\REL} \term$,
    $\term^{\set{1 \mapsto 1, 2 \mapsto 1}} \simalgclassx{\REL} (\term \cap \const{I}) \cdot \top$,
    $\term^{\set{1 \mapsto 2, 2 \mapsto 2}} \simalgclassx{\REL} \top \cdot (\term \cap \const{I})$, and
    $\term^{\set{1 \mapsto 2, 2 \mapsto 1}} = \term^{\smile}$.
    Hence, this completes the proof.
\end{proof}
\begin{remark}\label{remark: elim finite}
    $\const{D} \cdot \const{D} \notsimalgclassx{\REL} \top$, whereas $\const{D} \cdot \const{D} \simalgclassx{\REL_{\ge 3}} \top$.
    For example when $\card |\model| = 1$, since $\jump{\const{D}}_{\model} = \jump{\bot}_{\model}$,
    we have $\jump{\const{D} \cdot \const{D}}_{\model} = \emptyset \neq |\model| = \jump{\top}_{\model}$.
    ($\const{D} \cdot \const{D}$ is not equivalent to neither one of the four constants w.r.t.\  $\simalgclassx{\REL}$; thus, there are many constants w.r.t.\  $\simalgclassx{\REL}$.)
\end{remark}

\subsection{Monoid for $\voset{(\CoRSigma_{n})}{k}$} \label{section: CoR monoid}
Next, we decompose terms, and then we reduce the finiteness of $\voset{(\CoRSigma_{n})}{k}\quoset{\simalgclassx{\REL}}$ to that of a monoid (cf.\ \Cref{section: monoid}).
\begin{lemma}\label{lemma: 1 to k CoR}
    For each $n, k \in \nat$,
    if $\voset{(\CoRSigma_{n})}{1}\quoset{\simalgclassx{\REL}}$ is finite, $\voset{(\CoRSigma_{n})}{k}\quoset{\simalgclassx{\REL}}$ is finite.
\end{lemma}
\begin{proof}
    By specializing $\termset$ with $\voset{(\CoRSigma_{n})}{k}$ and $\mathcal{C}$ with $\REL$, in \Cref{proposition: finite to decidable} and \Cref{lemma: 1 to k}.
\end{proof}

\begin{lemma}\label{lemma: complement}
    For each $n, k \in \nat$,
    $\voset{(\CoRSigma_{n})}{k}\quoset{\simalgclassx{\REL}}$ is finite iff $\voset{(\CoRPi_{n})}{k}\quoset{\simalgclassx{\REL}}$ is finite.
\end{lemma}
\begin{proof}
    $\Longleftarrow$:
    For every term $\term$ in $\voset{(\CoRPi_{n})}{k}$,
    there is some $\term[2]$ in $\voset{(\CoRSigma_{n})}{k}$ such that $\term \simalgclassx{\REL} \term[2]^{-}$.
    Such $\term[2]$ can be obtained from the term $\term^{-}$ by taking the complement normal form using the following equations\ifiscameraready \else (see also \Cref{section: complement normal form})\fi:
    \begin{align*}
        \top^{-}                     & \sim_{\REL} \bot                           & \bot^{-}                     & \sim_{\REL} \top                           & \const{I}^{-}      & \sim_{\REL} \const{D} & \const{D}^{-}        & \sim_{\REL} \const{I}             \\
        (\term[2] \cup \term[3])^{-} & \sim_{\REL} \term[2]^{-} \cap \term[3]^{-} & (\term[2] \cap \term[3])^{-} & \sim_{\REL} \term[2]^{-} \cup \term[3]^{-} & (\term[2]^{-})^{-} & \sim_{\REL} \term[2]  & (\term[2]^{\pi})^{-} & \sim_{\REL} (\term[2]^{-})^{\pi}.
    \end{align*}
    $\Longrightarrow$:
    As with $\Longleftarrow$.
\end{proof}

\begin{lemma}[cf.\ \Cref{lemma: decomposition}]\label{lemma: decomposition n}
    Let $a \in \vsig$.
    For all $n \ge 2$, $\term \in \voset{(\CoRSigma_{n})}{1}$,
    there are $\term_0 \in \voset{(\CoRSigma_{1})}{1}$ and $\term_1 \in \voset{(\CoRPi_{n-1})}{1}$ such that $\term \simalgclassx{\REL_{\ge 3}} \term_0\assign{\term_1}{a}$.
\end{lemma}
\begin{proof}
    By induction on the pair of $n$ and $\term$.
    We distinguish the following cases.
    Case $\term \in \voset{(\CoRSigma_{n-1})}{1}$:
    Clear, by IH ($\because$ $\CoRPi_{n-2} \subseteq \CoRPi_{n-1}$).
    Case $\term \in \voset{(\CoRPi_{n-1})}{1}$:
    By letting $\term_0 = a$ and $\term_1 = \term$.
    Case $\term = \term[2] \cup \term[3]$:
    By $\vo(\term) = 1$, $\vo(\term[2]) = 0$ or $\vo(\term[3]) = 0$ holds.
    Sub-case $\vo(\term[2]) = 0$:
    By \Cref{lemma: CoR 0}, let $\term[2]' \in \set{\bot, \top, \const{I}, \const{D}}$ be s.t.\ $\term[2] \simalgclassx{\REL_{\ge 3}} \term[2]'$.
    By IH w.r.t.\ $\term[3]$, let $\term[3]_0 \in \voset{(\CoRSigma_{1})}{1}, \term[3]_1 \in \voset{(\CoRPi_{n-1})}{1}$ be s.t.\ $\term[3] \simalgclassx{\REL_{\ge 3}} \term[3]_0\assign{\term[3]_1}{a}$.
    By letting $\term_0 = \term[2]' \cup \term[3]_0$ and $\term_1 = \term[3]_1$, we have $\term \simalgclassx{\REL_{\ge 3}} \term[1]_0 \assign{\term[1]_1}{a}$.
    Sub-case $\vo(\term[3]) = 0$:
    As with Sub-case $\vo(\term[2]) = 0$.
    Case $\term = \term[2] \cap \term[3], \term[2] \cdot \term[3], \term[2]^{\pi}$:
    As with Case $\term = \term[2] \cup \term[3]$.
\end{proof}
The following is an illustrative example of the decomposition of \Cref{lemma: decomposition n}:

\noindent\scalebox{.69}{\parbox{1.428\linewidth}{
        \begin{align*}
            \begin{tikzpicture}[
                    level distance=4ex,
                    sibling distance=8ex,
                    baseline=-8ex
                ]
                \tikzset{fit1/.style={rounded rectangle, inner sep=2.pt, fill=yellow!40, opacity=.8},
                    fit2/.style={red!20, fill = red!20, thick},
                    fit3/.style={blue!20, fill = blue!20, thick},
                    nlab/.style={font= \scriptsize, color = gray}}
                \node(0){$\cdot$}
                child {node(00){$\const{D} \dagger \const{D}$}}
                child {node(01){$\cdot$}
                        child {node(010) {$\dagger$}
                                child {node(0100) {$\const{D} \cdot \const{D}$}
                                    }
                                child {node(0101) {$\dagger$}
                                        child {node(01010) {$\cdot$}
                                                child {node(010100) {$\const{D} \dagger \const{D}$}
                                                    }
                                                child {node(010101) {$a$}
                                                    }
                                            }
                                        child {node(01011) {$\const{D} \cdot \const{D}$}
                                            }
                                    }
                            }
                        child {node(011) {$\const{D} \cdot \const{D}$}
                            }
                    }
                ;
                \begin{pgfonlayer}{background}
                    \draw [fit3] ($(010.north) + (0,0)$) -- ($(010.south) + (-2, -1.8)$) -- ($(010.south) + (2, -1.8)$) -- cycle;
                    \node[fit1, rotate fit=-45, fit=(0)(01)] {};
                    \node[fit1, rotate fit=45, fit=(01)(010)] {};
                \end{pgfonlayer}
                \node[right = -2pt of 0, nlab]{$\CoRSigma_{3}$};
            \end{tikzpicture}
            \quad
            \simalgclassx{\REL_{\ge 3}}
            \quad
            \begin{tikzpicture}[
                    level distance=4ex,
                    sibling distance=8ex,
                    baseline=-8ex
                ]
                \tikzset{fit1/.style={rounded rectangle, inner sep=2.pt, fill=yellow!40, opacity=.8},
                    fit2/.style={red!20, fill = red!20, thick},
                    fit3/.style={blue!20, fill = blue!20, thick},
                    nlab/.style={font= \scriptsize, color = gray}}
                \node(0){$\cdot$}
                child {node(00){$\const{D}$}}
                child {node(01){$\cdot$}
                        child {node(010) {$\dagger$}
                                child {node(0100) {$\top$}
                                    }
                                child {node(0101) {$\dagger$}
                                        child {node(01010) {$\cdot$}
                                                child {node(010100) {$\const{D}$}
                                                    }
                                                child {node(010101) {$a$}
                                                    }
                                            }
                                        child {node(01011) {$\top$}
                                            }
                                    }
                            }
                        child {node(011) {$\top$}
                            }
                    }
                ;
                \begin{pgfonlayer}{background}
                    \draw [fit3] ($(010.north) + (0,0)$) -- ($(010.south) + (-2, -1.8)$) -- ($(010.south) + (2, -1.8)$) -- cycle;
                    \node[fit1, rotate fit=-45, fit=(0)(01)] {};
                    \node[fit1, rotate fit=45, fit=(01)(010)] {};
                \end{pgfonlayer}
                \node[right = -2pt of 0, nlab]{$\CoRSigma_{3}$};
                \node[right = -2pt of 01, nlab]{$\CoRSigma_{3}$};
                \node[right = -2pt of 010, nlab]{$\CoRPi_{2}$};
            \end{tikzpicture}
            \quad
            =
            \quad
            \left(
            \begin{tikzpicture}[
                        level distance=4ex,
                        sibling distance=6ex,
                        baseline=-4ex
                    ]
                    \tikzset{fit1/.style={rounded rectangle, inner sep=2.pt, fill=yellow!40, opacity=.8},
                        fit2/.style={red!20, fill = red!20, thick},
                        fit3/.style={blue!20, fill = blue!20, thick},
                        nlab/.style={font= \scriptsize, color = gray}}
                    \node(0){$\cdot$}
                    child {node(00){$\const{D}$}}
                    child {node(01){$\cdot$}
                            child {node(010) {$a$}
                                }
                            child {node(011) {$\top$}
                                }
                        }
                    ;
                    \begin{pgfonlayer}{background}
                        \node[fit1, rotate fit=-45, fit=(0)(01)] {};
                        \node[fit1, rotate fit=45, fit=(01)(010)] {};
                    \end{pgfonlayer}
                    \node[right = -2pt of 0, nlab]{$\CoRSigma_{1}$};
                \end{tikzpicture}\right)\left[\begin{tikzpicture}[
                                                      level distance=4ex,
                                                      sibling distance=6ex,
                                                      baseline=-6ex
                                                  ]
                                                  \tikzset{fit1/.style={rounded rectangle, inner sep=.8pt, fill=yellow!30},
                                                      fit2/.style={red!20, fill = red!20, thick},
                                                      fit3/.style={blue!20, fill = blue!20, thick},
                                                      nlab/.style={font= \scriptsize, color = red}}
                                                  \node(010) {$\dagger$}
                                                  child {node(0100) {$\top$}
                                                      }
                                                  child {node(0101) {$\dagger$}
                                                          child {node(01010) {$\cdot$}
                                                                  child {node(010100) {$\const{D}$}
                                                                      }
                                                                  child {node(010101) {$a$}
                                                                      }
                                                              }
                                                          child {node(01011) {$\top$}
                                                              }
                                                      }
                                                  ;
                                                  \begin{pgfonlayer}{background}
                        \draw [fit3] ($(010.north) + (0,0)$) -- ($(010.south) + (-2, -1.8)$) -- ($(010.south) + (2, -1.8)$) -- cycle;
                    \end{pgfonlayer}
                                                  \node[right = -2pt of 010, nlab]{$\CoRPi_{2}$};
                                              \end{tikzpicture} \quad / \quad a\right]
        \end{align*}}}
\begin{lemma}\label{lemma: dot-dagger 1 to n CoR}
    For each $n \in \nat$, if $\voset{(\CoRSigma_1)}{1}\quoset{\simalgclassx{\REL}}$ is finite, $\voset{(\CoRSigma_{n})}{1}\quoset{\simalgclassx{\REL}}$ is finite.
\end{lemma}
\begin{proof}
    By induction on $n$.
    Case $n \le 1$:
    By the assumption (note that $\voset{(\CoRSigma_0)}{1} \subseteq \voset{(\CoRSigma_1)}{1}$).
    Case $n \ge 2$:
    By the assumption, $\voset{(\CoRSigma_{1})}{1}\quoset{\simalgclassx{\REL}}$ is finite.
    By IH with \Cref{lemma: complement}, $\voset{(\CoRPi_{n-1})}{1}\quoset{\simalgclassx{\REL}}$ is finite.
    Combining them with \Cref{lemma: decomposition n} (and \Cref{proposition: eliminate fin} for changing $\simalgclassx{\REL}$ and $\simalgclassx{\REL_{\ge 3}}$ mutually) yields that $\voset{(\CoRSigma_{n})}{1}\quoset{\simalgclassx{\REL}}$ is finite.
\end{proof}
\newrobustcmd{\alltermsig}[1]{\kl[\alltermsig]{\mathbf{T}}^{#1}}%
\knowledge{\alltermsig}{automatic in command}%
For $\sig \subseteq \CoRsig$, let $\intro*\alltermsig{\sig} \subseteq \CoRallterm$ be the set of all terms over the signature $\sig$.
Then we have:
\begin{lemma}\label{lemma: simpl CoR}
    If $\voset{\alltermsig{\set{\cap, \cdot, \const{I}, \const{D}}}}{1}\quoset{\simalgclassx{\REL}}$ is finite,
    then $\voset{(\CoRSigma_1)}{1}\quoset{\simalgclassx{\REL}}$ is finite.
\end{lemma}
\begin{proof}[Proof sketch]
    Note that $\term[1], \term[2] \in \CoRSigma_1 \Coloneqq \term[3] \mid \term[1] \cup \term[2] \mid \term[1] \cap \term[2] \mid \term[1] \cdot \term[2] \mid \term[1]^{\pi}$ (where $\term[3] \in \CoRSigma_0$, $\pi \in \range{1, 2}^{\range{1, 2}}$)
    and $\term[3], \term[3]' \in \CoRSigma_0 \Coloneqq a \mid \term[3] \cup \term[3]' \mid \term[3] \cap \term[3]' \mid \term[3]^{-} \mid \top \mid \bot \mid \term[3]^{\pi}$ (where $a \in \vsig$, $\pi \in \range{1, 2}^{\range{1, 2}}$).
    By taking the complement ($-$) and projection ($\pi$) normal form and replacing $\bot$ with $\const{I} \cap \const{D}$ and $\top$ with $\const{I} \cup \const{D}$, for each $\term \in \voset{(\CoRSigma_1)}{1}$ and $a \in \vsig$,
    there are $\term_0 \in \voset{\alltermsig{\set{\cup, \cap, \cdot, \const{I}, \const{D}}}}{1}$ and $\term_1 \in \voset{\alltermsig{\set{-} \cup \set{\pi \mid \pi \in \range{1, 2}^{\range{1, 2}}}}}{1}$ such that $\term \simalgclassx{\REL} \term_0\assign{\term_1}{a}$.
    Moreover, by the distributive law of $\cup$ w.r.t.\ $\cdot$ and $\cap$, for each $\term \in \voset{\alltermsig{\set{\cup, \cap, \cdot, \const{I}, \const{D}}}}{1}$,
    there are $n \in \nat$ and $\term_1, \dots, \term_n \in \voset{\alltermsig{\set{\cap, \cdot, \const{I}, \const{D}}}}{1}$ such that
    $\term \simalgclassx{\REL} \term_1 \cup \dots \cup \term_n$.
    Because $\voset{\alltermsig{\set{\cap, \cdot, \const{I}, \const{D}}}}{1}\quoset{\simalgclassx{\REL}}$ is finite (by the assumption) and $\voset{\alltermsig{\set{-} \cup \set{\pi \mid \pi \in \range{1, 2}^{\range{1, 2}}}}}{1}\quoset{\simalgclassx{\REL}}$ is clearly finite, $\voset{(\CoRSigma_1)}{1}\quoset{\simalgclassx{\REL}}$ is finite.
    Hence, this completes the proof.
    (See \ifiscameraready \cite{nakamuraFiniteVariableOccurrence2023arxiv} \else \Cref{section: lemma: simple CoR} \fi for more details of the proof.)
\end{proof}
Combining \Cref{lemma: 1 to k CoR,lemma: dot-dagger 1 to n CoR,lemma: simpl CoR} yields that to prove that $\voset{(\CoRSigma_{n})}{k}\quoset{\simalgclassx{\REL}}$ is finite,
it suffices to prove that $\voset{\alltermsig{\set{\cap, \cdot, \const{I}, \const{D}}}}{1}\quoset{\simalgclassx{\REL}}$ is finite.

\knowledgenewrobustcmd\extpvsigx{\cmdkl{\dot{\Sigma}}}
\newrobustcmd{\extpvsigzerox}{\kl[\extpvsigzerox]{\dot{\Sigma}}_{\kl[\extpvsigzerox]{0}}}
\knowledge{\extpvsigzerox}{automatic in command}
\knowledgenewrobustcmd{\extpsimx}[1]{\mathbin{\cmdkl{\dot{\sim}}}_{#1}}

Let $\intro*\extpvsigx$ be the set of characters of \Cref{definition: unary maps} from the signature $\set{\cap^{(2)}, \cdot^{(2)}, \const{I}^{(0)}, \const{D}^{(0)}, \smile^{(1)}}$.
That is, $\intro*\extpvsigx \defeq \set{(\bl \cap \term), (\term \cap \bl), (\bl \cdot \term), (\term \cdot \bl) \mid \term \in \voset{(\alltermsig{\set{\cap, \cdot, \const{I}, \const{D}, \smile}})}{0}} \cup \set{\smile}$.
(While $\smile$ does not occur in $\alltermsig{\set{\cap, \cdot, \const{I}, \const{D}}}$, we introduce $\smile$ for replacing the primitive character $(\const{D} \cdot \bl)$ with $\smile (\bl \cdot \const{D})$. This is not essential but is useful for reducing the number of equations and for simplifying the notation (\Cref{definition: sig 0}).)
Let $\intro*\extpsimx{\REL_{\ge 5}}$ be the equivalence relation on $\extpvsigx^*$ defined by:
$\word[1] \extpsimx{\REL_{\ge 5}} \word[2] \defiff \word[1]\context{a} \simalgclassx{\REL_{\ge 5}} \word[2]\context{a}$ where $a \in \vsig$ is any variable (recall \Cref{definition: monoid equation}).\footnote{The condition ``$\ge 5$'' is needed for some equations in \Cref{figure: equations}\ifiscameraready \else (see also \Cref{section: lemma: soundness})\fi.}
\begin{lemma}\label{lemma: CoR monoid}
    If $\extpvsigx^{*}\quoset{\extpsimx{\REL_{\ge 5}}}$ is finite, then $\voset{\alltermsig{\set{\cap, \cdot, \const{I}, \const{D}}}}{1}\quoset{\simalgclassx{\REL}}$ is finite.
\end{lemma}
\begin{proof}
    Since $\extpvsigx^{*}\quoset{\extpsimx{\REL_{\ge 5}}}$ is finite,
    we have that $\voset{\alltermsig{\set{\cap, \cdot, \const{I}, \const{D}, \smile}}}{1}\quoset{\simalgclassx{\REL_{\ge 5}}}$ is finite (\Cref{lemma: monoid});
    thus, $\voset{\alltermsig{\set{\cap, \cdot, \const{I}, \const{D}}}}{1}\quoset{\simalgclassx{\REL_{\ge 5}}}$ is finite.
    Hence by \Cref{proposition: eliminate fin}, this completes the proof.
\end{proof}
We consider the following finite subset $\extpvsigzerox$ of $\extpvsigx$ (cf.\ \Cref{lemma: finite generator}):

\knowledgenewrobustcmd{\uI}{\cmdkl{\cup_{\const{I}}}}
\knowledgenewrobustcmd{\uD}{\cmdkl{\cup_{\const{D}}}}
\knowledgenewrobustcmd{\iI}{\cmdkl{\cap_{\const{I}}}}
\knowledgenewrobustcmd{\iD}{\cmdkl{\cap_{\const{D}}}}
\knowledgenewrobustcmd{\cD}{\cmdkl{\cdot_{\const{D}}}}
\knowledgenewrobustcmd{\re}{\cmdkl{\smile}}
\begin{definition}\label{definition: sig 0}
    Let $\extpvsigzerox \subseteq \extpvsigx$ be the finite set $\set{\intro*\iI, \intro*\iD, \intro*\cD, \intro*\re}$,
    where $\iI$, $\iD$, $\cD$ are abbreviations of $(\bl \cap \const{I})$, $(\bl \cap \const{D})$, $(\bl \cdot \const{D})$,
    respectively.
\end{definition}

\begin{lemma}\label{lemma: CoR monoid finite generator}
    If $\extpvsigzerox^{*}\quoset{\extpsimx{\REL_{\ge 5}}}$ is finite, then $\extpvsigx^{*}\quoset{\extpsimx{\REL_{\ge 5}}}$ is finite.
\end{lemma}
\begin{proof}
    It suffices to prove the following: for every $a \in \extpvsigx$, there is $\word \in \extpvsigzerox^*$ such that $a \extpsimx{\REL_{\ge 5}} \word$.
    Case $a = (\bl \cap \term), (\bl \cdot \term)$:
    Since $\vo(\term) = 0$, by using \Cref{lemma: CoR 0}, they are shown by distinguishing the following four sub-cases, as follows:
    \begin{center}
        \hfill
        \begin{tabular}{|c||c|c|c|c|}
            \hline                         & $\term \simalgclassx{\REL_{\ge 3}} \bot$ & $\term \simalgclassx{\REL_{\ge 3}} \top$ & $\term \simalgclassx{\REL_{\ge 3}} \const{I}$ & $\term \simalgclassx{\REL_{\ge 3}} \const{D}$ \\
            \hline
            \hline $a = (\bl \cap \term)$  & $\iI\iD$                                 & $\empword$                               & $\iI$                                         & $\iD$                                         \\
            \hline $a = (\bl \cdot \term)$ & $\iI\iD$                                 & $\cD\cD$                                 & $\empword$                                    & $\cD$                                         \\
            \hline                                                                                                                                                                          %
        \end{tabular}
        \hfill
        \phantom{.}
    \end{center}
    Case $a = (\term \cap \bl), (\term \cdot \bl)$:
    By $(\term \cap \bl) = \re (\bl \cap \term)$ and applying the above case analysis for $(- \cap \term)$, this case can be proved (similarly for $(\term \cdot \bl)$).
    Case $a = \re$:
    Since $\re \in \extpvsigzerox$.
\end{proof}
Thus, our goal is to prove that $\extpvsigzerox^{*}\quoset{\extpsimx{\REL_{\ge 5}}}$ is finite.

\subsection{On the finiteness of the monoid}\label{section: fin}
For the finiteness of $\extpvsigzero^{*}\quoset{\extpsimx{\REL_{\ge 5}}}$ (cf.\ \Cref{lemma: cofiniteness}),
we present the $21$ equations in \Cref{figure: equations}.\footnote{The most technical part of the paper is to collect these equations;
    they are obtained by running a program based on \Cref{algorithm: finiteness proof} using ATP/SMT systems.}
For $i \in \range{1, 21}$, let $\word[1]_i, \word[2]_i$ be words such that $\word[1]_i = \word[2]_i$ denotes the $i$-th equation.
\begin{figure}[t]
    \noindent
    \begin{minipage}[t]{0.25\textwidth}
        \begin{align}
            \iI    & = \iI\iI \\
            \iD    & = \iD\iD \\
            \iI    & = \iI\re \\
            \iI    & = \re\iI \\
            \iI\iD & = \iD\iI \\
            \iD\re & = \re\iD
        \end{align}
    \end{minipage}\hfill
    \begin{minipage}[t]{0.28\textwidth}
        \begin{align}
            \empword & = \re\re    \\
            \iD\iI   & = \cD\iI\iD \\
            \iD\iI   & = \iI\cD\iI \\
            \iI\cD   & = \iI\cD\iD \\
            \cD\iI   & = \iD\cD\iI \\
            \cD\cD   & = \cD\iD\cD
        \end{align}
    \end{minipage}\hfill
    \begin{minipage}[t]{0.41\textwidth}
        \begin{align}
            \label{equationa: example cDcD cDcDcD}   \cD\cD & = \cD\cD\cD       \\
            \cD\cD\iD                                       & = \cD\cD\iI\cD    \\
            \iD\cD\cD                                       & = \cD\iI\cD\cD    \\
            \re\cD\iI                                       & = \iD\re\cD\iI    \\
            \cD\re\cD\re                                    & = \re\cD\re\cD    \\
            \re\cD\iI\cD                                    & = \iD\re\cD\cD\iD \\
            \cD\re\cD\cD\re                                 & = \re\cD\cD\re\cD
        \end{align}
    \end{minipage}

    \begin{align}
        \cD\iD\re\cD\iD\re\cD\iD\re\cD\iD\re\cD\iD & = \re\cD\iD\re\cD\iD\re\cD\iD\re\cD\iD\re\cD\iD\re \\
        \cD\iI\cD\re\cD\iI\cD\re\cD                & = \re\cD\iI\cD\re\cD\iI\cD\re\cD
    \end{align}
    \caption{Equations for the finiteness}
    \label{figure: equations}
\end{figure}
\begin{lemma}[soundness]\label{lemma: soundness}
    For each $i \in \range{1, 21}$, $\word[1]_i \extpsimx{\REL_{\ge 5}} \word[2]_i$.
\end{lemma}
\begin{proof}[Proof Sketch]
    We prove $\word[1]_i\context{a} \simalgclassx{\REL_{\ge 5}} \word[2]_i\context{a}$, where $a$ is any variable.
    This equation can be translated to the validity of a first-order sentence via the standard translation \cite{Tarski1941}.
    Here, we add the formula $\exists x_1, \dots, x_5, \land_{i,j \in \range{1, 5}; i \neq j} x_i \neq x_j$ as an axiom, for forcing $\extpsimx{\REL_{\ge 5}}$.
    Thanks to this encoding, each of them can also be tested by using ATP/SMT systems.
    Nevertheless, in the following, as an example, we give explicit proof for \Cref{equationa: example cDcD cDcDcD}.
    By using the standard translation,
    \Cref{equationa: example cDcD cDcDcD} is translated into the following formula in first-order logic, where $x_0, y_0$ are free variables:
    \begin{align*}
                              & (\exists y_1, (\exists y_2, a(x_0, y_2) \land y_2 \neq y_1) \land y_1 \neq y_0)                                    \\
        \leftrightarrow \quad & (\exists y_1, (\exists y_2, (\exists y_3, a(x_0, y_3) \land y_3 \neq y_2) \land y_2 \neq y_1) \land y_1 \neq y_0).
    \end{align*}
    This formula is valid under $\REL_{\ge 5}$, which can be shown by using the axiom above (notice that under $\REL_{\ge 5}$, $y_1$ on the left and $y_1, y_2$ on the right always exist, by taking a vertex not assigned by any variable occurring in each formula;
    thus, both formulas are equivalent to the formula $\exists y, a(x_0, y)$).
    Even without the encoding to first-order logic, this equation can also be shown as follows:
    \begin{align*}
        (\cD \cD)\context{a} = (a \cdot \const{D}) \cdot \const{D} & = a \cdot (\const{D} \cdot \const{D})                    \tag{associativity law}                                                                                                        \\
                                                                   & = a \cdot \top                                           \tag{$\top \simalgclassx{\REL_{\ge 3}} \const{D} \cdot \const{D}$}                                                             \\
                                                                   & = a \cdot (\const{D} \cdot \const{D} \cdot \const{D})    \tag{$\top \simalgclassx{\REL_{\ge 3}} \top \cdot \const{D}$ and $\top \simalgclassx{\REL_{\ge 3}} \const{D} \cdot \const{D}$} \\
                                                                   & = ((a \cdot \const{D}) \cdot \const{D}) \cdot \const{D} = (\cD \cD \cD)\context{a}. \tag{associativity law}
    \end{align*}
    See \ifiscameraready \cite{nakamuraFiniteVariableOccurrence2023arxiv} \else \Cref{section: lemma: soundness} \fi for all the equations.
\end{proof}
\begin{lemma}\label{lemma: finite lang}
    The language $\bigcup_{i \in \range{1, 21}} \extpvsigzero^* \word[2]_i \extpvsigzero^*$ over $\extpvsigzero$ is \kl{cofinite}.
\end{lemma}
\begin{proof}[Proof Sketch]
    It suffices to prove that for some $n \in \nat$, the following hold:
    there is no word $\word \in \extpvsigzero^* \setminus (\bigcup_{i \in \range{1, 21}} \extpvsigzero^* \word[2]_i \extpvsigzero^*)$ such that $\wlen{\word} \ge n$ (since the set $\set{\word \in \extpvsigzero^* \mid \wlen{\word} \le n - 1}$ is finite).
    This holds when $n \ge 29$, which can be tested by using Z3 (an ATP/SMT system) \cite{demouraZ3EfficientSMT2008}
    and can be checked by drawing its DFA (see \ifiscameraready \cite{nakamuraFiniteVariableOccurrence2023arxiv}\else \Cref{appendix: SMT-LIB}\fi, for more details).
\end{proof}

Thus, we have obtained the following:
\begin{lemma}\label{lemma: result}
    $\extpvsig^*\quoset{\extpsimx{\REL_{\ge 5}}}$ is finite.
\end{lemma}
\begin{proof}
    By \Cref{lemma: soundness,lemma: finite lang}, we can apply \Cref{lemma: cofiniteness},
    where $<$ is the shortlex order on $\extpvsigzero^*$ induced by: $\iI < \iD < \cD < \re$.
    By the form, $\word[1]_i < \word[2]_i$ is clear for each $i \in \range{1, 21}$.
\end{proof}
Finally, \Cref{theorem: main} is obtained as follows:
\begin{proof}[Proof of \Cref{theorem: main}]
    We have:
    $\extpvsigzero^*\quoset{\extpsimx{\REL_{\ge 5}}}$ is finite (\Cref{lemma: result})
    $\Longrightarrow$ $\extpvsig^*\quoset{\extpsimx{\REL_{\ge 5}}}$ is finite (\Cref{lemma: CoR monoid finite generator})
    $\Longrightarrow$ $\voset{\allterm^{\set{\cap, \cdot, \const{I}, \const{D}}}}{1}\quoset{\simalgclassx{\REL}}$ is finite (\Cref{lemma: CoR monoid})
    $\Longrightarrow$ $\voset{(\CoRSigma_1)}{1}\quoset{\simalgclassx{\REL}}$ is finite (\Cref{lemma: simpl CoR})
    $\Longrightarrow$ $\voset{(\CoRSigma_n)}{1}\quoset{\simalgclassx{\REL}}$ is finite (\Cref{lemma: dot-dagger 1 to n CoR})
    $\Longrightarrow$ $\voset{(\CoRSigma_n)}{k}\quoset{\simalgclassx{\REL}}$ is finite (\Cref{lemma: 1 to k CoR}).
\end{proof}

\begin{remark}
    The finite axiomatizability of the \kl{equational theory} of $\voset{(\CoRSigma_n)}{k}$ over $\REL$ immediately follows from the finiteness of $\voset{(\CoRSigma_n)}{k}\quoset{\simalgclassx{\REL}}$.
\end{remark}
\section{Conclusion}\label{section: conclusion}
We have introduced the \emph{\kl{$k$-variable-occurrence fragment}}
and presented an approach for showing the decidability of the equational theory from the finiteness.
As a case study, we have proved that the \kl{equational theory} of $\voset{(\CoRSigma_n)}{k}$ is decidable,
whereas that of $\CoRSigma_n$ is undecidable in general.
We leave the decidability open for the \kl{equational theory} of \kl{CoR} with \emph{full \kl{dot-dagger alternation}} (i.e., $\voset{\CoRallterm}{k}$, in this paper).
Our approach may apply to some other algebras/logics.
It would be interesting to consider the finite variable-occurrence fragment for other systems (e.g., \kl{CoR} with antidomain \cite{Hollenberg1997, Desharnais2009}, dynamic logics \cite{Harel2000}).
It would also be interesting to extend our result to \emph{first-order logic with equality} (cf.\ \Cref{proposition: CoR and FO3})---for example, is the \emph{$k$-atomic-predicate-occurrence fragment} of the $m$ variable fragment of first-order logic with equality decidable?
\bibliography{main-pand}
\ifiscameraready
\else
    \appendix

\section{Proof of \Cref{proposition: decidable/undecidable}}\label{section: proposition: decidable/undecidable}
(In this section, we refer, e.g.,\ \cite{Borger1997}, for standard terminologies in first-order logic.)
First, we show \Cref{proposition: decidable/undecidable} when $\vsig$ is a countably infinite set:
\begin{proposition}[\Cref{proposition: decidable/undecidable} for countably infinite set $\vsig$]\label{proposition: decidable/undecidable infinite}
    When $\vsig$ is a countably infinite set,
    the \kl{equational theory} of $\CoRSigma_{n}$ (resp.\ $\CoRPi_{n}$) is decidable if $n \le 1$ and is undecidable if $n \ge 2$.
\end{proposition}
\begin{proof}
    It suffices to consider $\CoRSigma_{n}$.
    Case $n \le 1$:
    Let $\term[1], \term[2] \in \CoRSigma_{1}$.
    Let $z_1, z_2$ be distinct variables of first-order formulas.
    By using the recursive translation of \Cref{proposition: dot-dagger} \cite{nakamuraExpressivePowerSuccinctness2022},
    let $\fml[1]$ (resp.\ $\fml[2]$) be the formula in the level $\Sigma_{1}$ of the three variable fragment of first-order logic with equality with two free variables $z_1, z_2$
    such that (w.r.t.\ binary relations) $\lambda z_1 z_2. \fml[1]$ (resp.\ $\lambda z_1 z_2. \fml[2]$)\footnote{We use $\lambda z_1 z_2. \fml[1]$ to denote the term having the semantics $\jump{\lambda z_1 z_2. \fml[1]}_{\model} \defeq \set{\tuple{x_1, x_2} \in |\model|^2 \mid \mbox{$\fml$ holds on $\model$ with all valuations mapping $z_1$ into $x_1$ and $z_2$ into $x_2$}}$ on a structure $\model$. (In \cite{nakamuraExpressivePowerSuccinctness2022}, the notation ``$[\fml]_{z_1 z_2}$'' is used instead of $\lambda z_1 z_2. \fml[1]$.)} is semantically equivalent to $\term[1]$ (resp.\ $\term[2]$) over $\REL$.
    Then, $\term[1] \sim_{\REL} \term[2]$ $\iff$ the formula $\forall z_1 z_2, \fml[1] \leftrightarrow \fml[2]$ is valid over $\REL$.
    By taking the prenex normal form\footnote{$\REL$ does not contain the empty structure; thus we can take the prenex normal form.},
    let $\fml[1] \sim_{\REL} \exists x_1 \dots x_n, \fml[1]_0$ (resp.\ $\fml[2]  \sim_{\REL} \exists y_1 \dots y_m, \fml[2]_0$), where $\fml[1]_0$ and $\fml[2]_0$ are quantifier-free formulas.
    Then,
    \begin{align*}
        \term[1] \not\sim_{\REL} \term[2]
         & \iff \lnot (\forall z_1 z_2, \fml[1] \leftrightarrow \fml[2]) \mbox{ is satisfiable}                                                     \\
         & \iff (\exists z_1 z_2, \fml[1] \land \lnot \fml[2]) \lor (\exists z_1 z_2, \lnot \fml[1] \land \fml[2]) \mbox{ is satisfiable}           \\
         & \iff (\exists z_1 z_2, (\exists x_1 \dots x_n, \fml_0) \land (\forall y_1 \dots y_m, \lnot \fml[2]_0))                                   \\
         & \qquad \lor (\exists z_1 z_2, (\exists y_1 \dots y_m, \lnot \fml[1]_0) \land (\forall x_1 \dots x_n, \fml[2]_0)) \mbox{ is satisfiable}.
    \end{align*}
    By taking the prenex normal form, the last sentence is equivalent to a sentence in the \kl{BSR class} $[\exists^*\forall^*, \mathrm{all}]_{=}$.
    Because the satisfiability problem of the \kl{BSR class} is decidable \cite{Bernays1928,Ramsey1930},
    the \kl{equational theory} of $\CoRSigma_{1}$ over $\REL$ is decidable.

    Case $n \ge 2$:
    The class $[\forall \exists \land \forall^3, (\omega, 1)]$ (hence, $[\forall \exists \land \forall^3, (0, \omega)]$) is a conservative reduction class \cite[Cor.\  3.1.19]{Borger1997}; thus, the satisfiability problem of the class $[\forall \exists \land \forall^3, (0, \omega)]$ is undecidable.
    Let $(\forall x, \exists y, \fml[1]) \land (\forall x y z, \fml[2])$ be a sentence in the class $[\forall \exists \land \forall^3, (0, \omega)]$, where $\fml[1], \fml[2]$ are quantifier-free.
    Since $(\forall x y z, \fml[2])$ and $(\exists x, \forall y, \lnot \fml[1])$ are in the level $\Sigma_{2}$ of the three variable fragment of first-order logic,
    by using the recursive translation of \Cref{proposition: dot-dagger} \cite{nakamuraExpressivePowerSuccinctness2022},
    let $\term[1]$ (resp.\ $\term[2]$) be the term in $\CoRSigma_{2}$
    such that $\lambda z_1 z_2. (\forall x y z, \fml[2])$ (resp.\ $\lambda z_1 z_2. \exists x, \forall y, \lnot \fml[1]$) is semantically equivalent to $\term[1]$ (resp.\ $\term[2]$) over $\REL$, where $z_1, z_2$ are any pairwise distinct variables.
    Then, we have:
    \begin{align*}
         & (\forall x, \exists y, \fml[1]) \land (\forall x y z, \fml[2]) \mbox{ is satisfiable} \iff \lnot((\forall x, \exists y, \fml[1]) \land (\forall x y z, \fml[2])) \mbox{ is valid} \\
         & \iff (\forall x y z, \fml[2]) \to (\exists x, \forall y, \lnot \fml[1]) \mbox{ is valid}                                                                                          \\
         & \iff ((\forall x y z, \fml[2]) \lor (\exists x, \forall y, \lnot \fml[1])) \leftrightarrow (\exists x, \forall y, \lnot \fml[1]) \mbox{ is valid}                                 \\
         & \iff \term[1] \cup \term[2] \sim_{\REL} \term[2].
    \end{align*}
    Thus, as we can give a reduction from the satisfiability problem of the class $[\forall \exists \land \forall^3, (0, \omega)]$,
    the \kl{equational theory} of $\CoRSigma_{2}$ over $\REL$ is undecidable.
\end{proof}

\subsection{Proof of \Cref{proposition: decidable/undecidable} when $\vsig$ is finite}

\begin{lemma}\label{lemma: elim converse}
    Assume that $\vsig$ is a countably infinite set.
    Then, the \kl{equational theory} over $\REL$ is undecidable
    for \emph{complement-free $\CoRPi_{2}$} (i.e., the class of terms in $\CoRPi_{2}$ not having the complement operator ($-$)).
\end{lemma}
\begin{proof}
    We give a reduction from that \kl{equational theory} of $\CoRPi_{2}$ (which is undecidable by \Cref{proposition: decidable/undecidable}).
    By taking the complement normal form (cf.\ \Cref{lemma: complement normal form}),
    it suffices to give a reduction from that \kl{equational theory} of terms in $\CoRPi_{2}$ such that the complement operator ($-$) only applies to variables.
    Let $\term[1], \term[2]$ be such terms.
    Let $\vsig_{\set{\term[1], \term[2]}}$ be the finite set of all variables occurring in $\term[1]$ or $\term[2]$.
    Let $\bar{\bullet} \colon \vsig_{\set{\term[1], \term[2]}} \to (\vsig \setminus \vsig_{\set{\term[1], \term[2]}})$ be an injective map.
    Let $\term[1]', \term[2]', \term[3], \term[3]'$ be the terms defined as follows:
    \begin{align*}
        \term[1]' & \defeq \mbox{ the term $\term[1]$ in which each occurrence $a^{-}$ has been replaced with the variable $\bar{a}$} \\
        \term[2]' & \defeq \mbox{ the term $\term[2]$ in which each occurrence $a^{-}$ has been replaced with the variable $\bar{a}$} \\
        \term[3]  & \defeq \bot \dagger (\bigcap_{a \in \vsig_{\set{\term[1], \term[2]}}} (a \cup \bar{a})) \dagger \bot              \\
        \term[3]' & \defeq \top \cdot (\bigcup_{a \in \vsig_{\set{\term[1], \term[2]}}} (a \cap \bar{a})) \cdot \top.
    \end{align*}
    ($\term[3]$ and $\term[3]'$ are used to force $\jump{a}_{\model} \cup \jump{\bar{a}}_{\model} = |\model|^2$ and $\jump{a}_{\model} \cap \jump{\bar{a}}_{\model} = \emptyset$, respectively.)
    Then, we have:
    \begin{align*}
        \term[1] \not\sim_{\REL} \term[2] & \iff (\term[1]' \cap \term[3]) \cup \term[3]' \not\sim_{\REL} (\term[2]' \cap \term[3]) \cup \term[3]'.
    \end{align*}
    $\Longrightarrow$:
    Let $\model$ be s.t.\ $\jump{\term[1]}_{\model} \neq \jump{\term[2]}_{\model}$.
    Let $\model'$ be the structure $\model$ in which $\bar{a}^{\model'} \defeq |\model'|^2 \setminus a^{\model'}$.
    Then, $\jump{\term[1]}_{\model} = \jump{\term[1]'}_{\model'}$ and $\jump{\term[2]}_{\model} = \jump{\term[2]'}_{\model'}$ hold by $\jump{a^{-}}_{\model} = \jump{\bar{a}}_{\model'}$.
    Also by $\bar{a}^{\model'} = |\model'|^2 \setminus a^{\model'}$, $\jump{\term[3]}_{\model'} = |\model'|^2$ and $\jump{\term[3]'}_{\model'} = \emptyset$ hold.
    Thus, $\jump{(\term[1]' \cap \term[3]) \cup \term[3]'}_{\model'} = \jump{\term[1]'}_{\model'} = \jump{\term[1]}_{\model} \neq \jump{\term[2]}_{\model} = \jump{\term[2]'}_{\model'} = \jump{(\term[2]' \cap \term[3]) \cup \term[3]'}_{\model'}$.
    $\Longleftarrow$:
    Let $\model$ be s.t.\ $\jump{(\term[1]' \cap \term[3]) \cup \term[3]'}_{\model} \neq \jump{(\term[2]' \cap \term[3]) \cup \term[3]'}_{\model}$.
    Then, we have $\jump{a \cap \bar{a}}_{\model} = \emptyset$ for $a \in \vsig_{\set{\term[1], \term[2]}}$
    ($\because$ If we assume $\jump{a \cap \bar{a}}_{\model} \neq \emptyset$,
    then $\jump{\term[3]'}_{\model} = |\model|^2$, and thus $\jump{(\term[1]' \cap \term[3]) \cup \term[3]'}_{\model} = |\model|^2 = \jump{(\term[2]' \cap \term[3]) \cup \term[3]'}_{\model}$;
    so reaching a contradiction); thus $\jump{\term[3]'}_{\model} = \emptyset$.
    Similarly, we have $\jump{a \cup \bar{a}}_{\model} = |\model|^2$ for $a \in \vsig_{\set{\term[1], \term[2]}}$
    ($\because$ If we assume $\jump{a \cup \bar{a}}_{\model} \neq |\model|^2$;
    then $\jump{\term[3]}_{\model} = \emptyset$, and thus $\jump{(\term[1]' \cap \term[3]) \cup \term[3]'}_{\model} = \emptyset = \jump{(\term[2]' \cap \term[3]) \cup \term[3]'}_{\model}$;
    so reaching a contradiction); thus $\jump{\term[3]}_{\model} = |\model|^2$.
    By $\jump{a \cap \bar{a}}_{\model} = \emptyset$ and $\jump{a \cup \bar{a}}_{\model} = |\model|^2$, we have $\jump{a^{-}}_{\model} = \jump{\bar{a}}_{\model}$,
    and thus $\jump{\term[1]}_{\model} = \jump{\term[1]'}_{\model}$ and $\jump{\term[2]}_{\model} = \jump{\term[2]'}_{\model}$.
    Thus, $\jump{\term[1]}_{\model} = \jump{\term[1]'}_{\model} = \jump{(\term[1]' \cap \term[3]) \cup \term[3]'}_{\model} \neq \jump{(\term[2]' \cap \term[3]) \cup \term[3]'}_{\model} = \jump{\term[2]'}_{\model} = \jump{\term[2]}_{\model}$.

    Finally, because $\term[1]', \term[2]' \in \CoRPi_{2}$, $\term[3] \in \CoRPi_{1}$, $\term[3]' \in \CoRSigma_{1}$,
    we have $((\term[1]' \cap \term[3]) \cup \term[3]'), ((\term[2]' \cap \term[3]) \cup \term[3]') \in \CoRPi_{2}$.
    Hence, this completes the proof.
\end{proof}

\newcommand{\trans}{\mathcal{T}}
\begin{lemma}\label{lemma: elim converse 2}
    If $\vsig$ is a non-empty set,
    the \kl{equational theory} of $\CoRSigma_{2}$ is undecidable.
\end{lemma}
\begin{proof}
    Let $\vsig' = \set{a_1, a_2, \dots}$ be a countably infinite set.
    We give a reduction from the \kl{equational theory} of complement-free $\CoRPi_{2}$ over $\vsig'$ (cf.\ \Cref{lemma: elim converse}).
    Let $a^{\Delta} \defeq (a^{\set{1 \mapsto 1, 2 \mapsto 1}} \cap a^{\set{1 \mapsto 2, 2 \mapsto 2}})$;
    note that $\jump{a^{\Delta}}_{\model} = \set{\tuple{x, y} \mid \tuple{x, x}, \tuple{y,y} \in \jump{a}_{\model}}$ for any $\model$.
    Let $\trans(\term[1])$ be the term defined by:
    \begin{align*}
        \trans(a_i)                       & \defeq a^{\Delta} \cap (a \cdot ((a^{-} \cap \const{I}) \cdot a)^{i})      \mbox{\quad  for $i \ge 1$}              &
        \trans(\term[2]^{\pi})            & \defeq a^{\Delta} \cap \trans(\term[2])^{\pi}                                                                         \\
        \trans(\top)                      & \defeq a^{\Delta}                                                                                                   &
        \trans(\term[2] \cup \term[3])    & \defeq \trans(\term[2]) \cup \trans(\term[3])                                                                         \\
        \trans(\bot)                      & \defeq \bot                                                                                                         &
        \trans(\term[2] \cap \term[3])    & \defeq \trans(\term[2]) \cap \trans(\term[3])                                                                         \\
        \trans(\const{I})                 & \defeq a^{\Delta} \cap \const{I}                                                                                    &
        \trans(\term[2] \cdot \term[3])   & \defeq \trans(\term[2]) \cdot \trans(\term[3])                                                                        \\
        \trans(\const{D})                 & \defeq a^{\Delta} \cap \const{D}                                                                                    &
        \trans(\term[2] \dagger \term[3]) & \defeq a^{\Delta} \cap (((a^{\Delta})^{-} \cup \trans(\term[2])) \dagger ((a^{\Delta})^{-} \cup \trans(\term[3]))).
    \end{align*}
    (Cf.\ \cite[Lem.\ 11]{Nakamura2019}.
    This construction can be viewed as a variant of the \emph{relativization} in logic---the translations $(\exists x, \fml) \leadsto (\exists x, A(x) \land \fml)$ and
    $(\forall x, \fml) \leadsto (\forall x, A(x) \to \fml$).
    The construction above is refined not to increase the dot-dagger alternation.)
    \begin{claim}\label{claim: preserve k}
        For every $\term[1]$ such that the complement operator does not occur in $\term[1]$,
        \begin{itemize}
            \item  if $\term[1] \in \CoRPi_0$, then $\trans(\term[1]) \in \CoRSigma_1$;
            \item  if $\term[1] \in \CoRSigma_1$, then $\trans(\term[1]) \in \CoRSigma_1$;
            \item  if $\term[1] \in \CoRPi_2$, then $\trans(\term[1]) \in \CoRPi_2$.
        \end{itemize}
    \end{claim}
    \begin{claimproof}
        By straightforward induction on $\term[1]$ using $\trans(a_i) \in \CoRSigma_{1}$.
        Note that $a^{\Delta} \in \CoRSigma_0 (= \CoRPi_0)$.
    \end{claimproof}
    For a structure $\model$, let $\model^{\mathrm{t}}$ be the structure defined by:
    \begin{align*}
        |\model^{\mathrm{t}}| & \defeq \set{x \mid \tuple{x, x} \in a^{\model}} & a_i^{\model^{\mathrm{t}}} & \defeq \jump{\trans(a_i)}_{\model} \mbox{\quad ($a_i \in \vsig'$)}.
    \end{align*}
    Then, the following two hold:
    \begin{claim}\label{claim: elim converse 2 faithful}
        For every $\model$ and $\term$, $\jump{\term}_{\model^{\mathrm{t}}} = \jump{\trans(\term)}_{\model}$.
    \end{claim}
    \begin{claimproof}
        By induction on $\term$.
        Case $\term = \top$:
        \begin{align*}
            \jump{\top}_{\model^{\mathrm{t}}} = |\model^{\mathrm{t}}|^2 & = \jump{a^{\Delta}}_{\model} \tag{By Def.\ of $\model^{\mathrm{t}}$} \\
                                                                        & = \jump{\trans(\top)}_{\model}.
        \end{align*}
        Case $\term = \bot, \const{I}, \const{D}$:
        As with the case of $\term = \top$.
        Case $\term = \term[2] \cdot \term[3]$:
        \begin{align*}
            \jump{\term[2] \cdot \term[3]}_{\model^{\mathrm{t}}} = \jump{\term[2]}_{\model^{\mathrm{t}}} \cdot \jump{\term[3]}_{\model^{\mathrm{t}}} & = \jump{\trans(\term[2])}_{\model} \cdot \jump{\trans(\term[3])}_{\model}       \tag{IH} \\
                                                                                                                                                     & = \jump{\trans(\term[2])\cdot \trans(\term[3])}_{\model}                                 \\
                                                                                                                                                     & = \jump{\trans(\term[2] \cdot \term[3])}_{\model}.
        \end{align*}
        Case $\term = \term[2] \cup \term[3], \term[2] \cap \term[3]$:
        As with the case of $\term = \term[2] \cdot \term[3]$.
        Case $\term = \term[2] \dagger \term[3]$ (note that the operator $\dagger$ depends on the universe):
        Note that if $\jump{\term[2]}_{\model^{\mathrm{t}}} = \jump{\trans(\term[2])}_{\model}$, we have the following:
        \begin{align*}
            \jump{\term[2]^{-}}_{\model^{\mathrm{t}}} = |\model^{\mathrm{t}}|^2 \setminus \jump{\term[2]}_{\model^{\mathrm{t}}} & =  |\model^{\mathrm{t}}|^2 \setminus \jump{\trans(\term[2])}_{\model} \tag{By assumption} \\
                                                                                                                                & = |\model^{\mathrm{t}}|^2 \cap \jump{\trans(\term[2])^{-}}_{\model}                       \\
                                                                                                                                & = \jump{a^{\Delta} \cap \trans(\term[2])^{-}}_{\model}.
        \end{align*}
        By using this, we have:
        \begin{align*}
            \jump{\term[2] \dagger \term[3]}_{\model^{\mathrm{t}}} & = |\model^{\mathrm{t}}|^2 \setminus (\jump{\term[2]^{-}}_{\model^{\mathrm{t}}} \cdot \jump{\term[3]^{-}}_{\model^{\mathrm{t}}})                                                        \\
                                                                   & = |\model^{\mathrm{t}}|^2 \setminus ( \jump{a^{\Delta} \cap \trans(\term[2])^{-}}_{\model} \cdot  \jump{a^{\Delta} \cap \trans(\term[3])^{-}}_{\model})     \tag{By IH with the above} \\
                                                                   & = \jump{a^{\Delta} \cap ((a^{\Delta} \cap \trans(\term[2])^{-}) \cdot (a^{\Delta} \cap \trans(\term[3])^{-}))^{-}}_{\model}                                                            \\
                                                                   & = \jump{a^{\Delta} \cap (((a^{\Delta})^{-} \cup \trans(\term[2])) \dagger ((a^{\Delta})^{-} \cup \trans(\term[3])))}_{\model}                                                          \\
                                                                   & = \jump{\trans(\term[2] \dagger \term[3])}_{\model}.
        \end{align*}
        Hence, this completes the proof.
    \end{claimproof}
    \begin{claim}[cf.\ {\cite[Lem.\ 11]{Nakamura2019}}]\label{claim: elim converse 2 surjective}
        For every $\model$, there is a structure $\model[2]$ such that $\model = \model[2]^{\mathrm{t}}$.
    \end{claim}
    \begin{claimproof}
        Let $\model[2] = \tuple{|\model[2]|, \set{a^{\model[2]}}_{a \in \set{a}}}$ be the structure defined by:
        \begin{align*}
            |\model[2]|   & \defeq |\model| \cup \set{\tuple{x, y, i, j} \mid \tuple{x, y} \in \jump{a_i}_{\model}, j \in \range{1, i}};                       \\
            a^{\model[2]} & \defeq \const{I}_{|\model|} \cup \set{\tuple{x, \tuple{x, y, i, 1}}, \tuple{\tuple{x, y, i, i}, y} \mid
            \tuple{x, y} \in \jump{a_i}_{\model}}                                                                                                              \\
                          & \quad \cup \set{\tuple{\tuple{x, y, i, j-1}, \tuple{x, y, i, j}} \mid   \tuple{x, y} \in \jump{a_i}_{\model}, j \in \range{2, i}}.
        \end{align*}
        Here, we assume that $|\model|$ and $\set{\tuple{x, y, i, j} \mid \tuple{x, y} \in \jump{a_i}_{\model}, j \in \range{1, i}}$ are disjoint (the other case can be shown in the same manner by renaming vertices).
        The following is an illustration of an example of the conversion from $\model[1]$ to $\model[2]$:\\
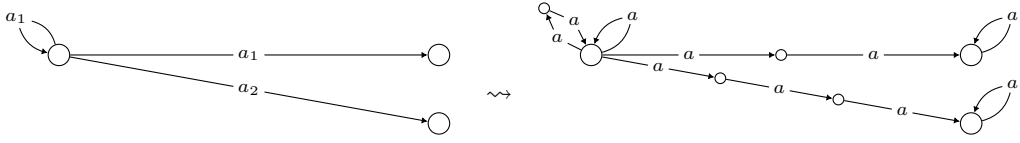
        \begin{tikzpicture}[
                baseline=-4ex
            ]
            \tikzset{fit1/.style={rounded rectangle, inner sep=0.8pt, fill=yellow!30},
                fit2/.style={red!20, fill = red!20, thick},
                fit3/.style={blue!20, fill = blue!20, thick},
                nlab/.style={font= \scriptsize, color = gray}}
            \graph[grow right = 5.0cm, branch down = 6ex, nodes={inner sep = .1cm}]{
            {0/{}[draw, circle]}-!-{1/{}[draw, circle], 2/{}[draw, circle]}
            };
            \path (0) edge [draw = white, opacity = 0] node[pos = .5, inner sep = .05cm](011){} (1);
            \path (0) edge [draw = white, opacity = 0] node[pos = .33, inner sep = .05cm](021){} (2);
            \path (0) edge [draw = white, opacity = 0] node[pos = .66, inner sep = .05cm](022){} (2);
            \path (0) edge [draw = white, opacity = 0] node[pos= 0.5, inner sep = 1.5pt, font = \scriptsize, opacity = 1](a1){$a_1$}(1);
            \path (0) edge [draw = white, opacity = 0, above left, out=150, in=120, loop] node[pos= 0.5, inner sep = 1.5pt, font = \scriptsize, opacity = 1](a2){$a_1$}(0);
            \path (0) edge [draw = white, opacity = 0] node[pos= 0.5, inner sep = 1.5pt, font = \scriptsize, opacity = 1](a3){$a_2$}(2);
            \graph[use existing nodes, edges={color=black, pos = .5, earrow}, edge quotes={fill=white, inner sep=1pt,font= \scriptsize}]{
            0 -- a1 -> 1;
            0 --[bend right] a2 ->[bend right] 0;
            0 -- a3 -> 2;
            };
        \end{tikzpicture}
        \quad
        $\leadsto$
        \begin{tikzpicture}[
                baseline=-4ex
            ]
            \tikzset{fit1/.style={rounded rectangle, inner sep=0.8pt, fill=yellow!30},
                fit2/.style={red!20, fill = red!20, thick},
                fit3/.style={blue!20, fill = blue!20, thick},
                nlab/.style={font= \scriptsize, color = gray}}
            \graph[grow right = 5.0cm, branch down = 6ex, nodes={inner sep = .1cm}]{
            {0/{}[draw, circle]}-!-{1/{}[draw, circle], 2/{}[draw, circle]}
            };
            \path (0) edge [draw = white, opacity = 0] node[pos = .5, inner sep = .05cm, draw=black, circle, opacity = 1](011){} (1);
            \path (0) edge [draw = white, opacity = 0] node[pos = .33, inner sep = .05cm, draw=black, circle, opacity = 1](021){} (2);
            \path (0) edge [draw = white, opacity = 0] node[pos = .66, inner sep = .05cm, draw=black, circle, opacity = 1](022){} (2);
            \path (0) edge [draw = white, opacity = 0, above left, out=150, in=120, loop, looseness=30] node[pos= 0.5, inner sep = .05cm, font = \scriptsize, draw=black, circle, opacity = 1](001){}(0);
            \path (0) edge [draw = white, opacity = 0, above right, out=30, in=60, loop, looseness=20] node[pos = 0.5, inner sep = 1.5pt, font = \scriptsize, opacity = 1](0l){$a$}(0);
            \path (1) edge [draw = white, opacity = 0, above right, out=30, in=60, loop, looseness=20] node[pos = 0.5, inner sep = 1.5pt, font = \scriptsize, opacity = 1](1l){$a$}(1);
            \path (2) edge [draw = white, opacity = 0, above right, out=30, in=60, loop, looseness=20] node[pos = 0.5, inner sep = 1.5pt, font = \scriptsize, opacity = 1](2l){$a$}(2);
            \path (0) edge [draw = white, opacity = 0] node[pos= 0.5, inner sep = 1.5pt, font = \scriptsize, opacity = 1](a11){$a$}(011);
            \path (011) edge [draw = white, opacity = 0] node[pos= 0.5, inner sep = 1.5pt, font = \scriptsize, opacity = 1](a12){$a$}(1);
            \path (0) edge [draw = white, opacity = 0, bend left] node[pos= 0.5, inner sep = 1.5pt, font = \scriptsize, opacity = 1](a21){$a$}(001);
            \path (001) edge [draw = white, opacity = 0, bend left] node[pos= 0.5, inner sep = 1.5pt, font = \scriptsize, opacity = 1](a22){$a$}(0);
            \path (0) edge [draw = white, opacity = 0] node[pos= 0.5, inner sep = 1.5pt, font = \scriptsize, opacity = 1](a31){$a$}(021);
            \path (021) edge [draw = white, opacity = 0] node[pos= 0.5, inner sep = 1.5pt, font = \scriptsize, opacity = 1](a32){$a$}(022);
            \path (022) edge [draw = white, opacity = 0] node[pos= 0.5, inner sep = 1.5pt, font = \scriptsize, opacity = 1](a33){$a$}(2);
            \graph[use existing nodes, edges={color=black, pos = .5, earrow}, edge quotes={fill=white, inner sep=1pt,font= \scriptsize}]{
            0 -- a11 -> 011; 011 -- a12 -> 1;
            0 -- a21 -> 001; 001 -- a22 -> 0;
            0 -- a31 -> 021; 021 -- a32 -> 022; 022 -- a33 -> 2;
            0 --[bend right] 0l ->[bend right] 0;
            1 --[bend right] 1l ->[bend right] 1;
            2 --[bend right] 2l ->[bend right] 2;
            };
        \end{tikzpicture}

        \noindent By the construction of $\model[2]$, we have:
        \begin{align*}
            \set{x \mid \tuple{x, x} \in a^{\model[2]}} & = |\model|      &
            \jump{\mathcal{T}(a_i)}_{\model[2]}         & = a_i^{\model}.
        \end{align*}
        Thus we have $|\model[2]^{\mathrm{t}}| = |\model|$ and $a_i^{\model[2]^{\mathrm{t}}} = a_i^{\model}$.
        Hence, $\model[2]^{\mathrm{t}} = \model$.
    \end{claimproof}
    From these two, we have:
    \begin{align*}
        \term[1] \not\sim_{\REL} \term[2] & \iff \jump{\term[1]}_{\model} \not\sim_{\REL} \jump{\term[2]}_{\model} \mbox{ for some $\model$}                                                                                                                                                                          \\
                                          & \iff \jump{\term[1]}_{\model[2]^{\mathrm{t}}} \not\sim_{\REL} \jump{\term[2]}_{\model[2]^{\mathrm{t}}} \mbox{ for some $\model[2]$}  \tag{$\Longrightarrow$: \Cref{claim: elim converse 2 surjective}. $\Longleftarrow$: By letting $\model[1] = \model[2]^{\mathrm{t}}$} \\
                                          & \iff \jump{\trans(\term[1])}_{\model[2]} \not\sim_{\REL} \jump{\trans(\term[2])}_{\model[2]} \mbox{ for some $\model[2]$} \tag{\Cref{claim: elim converse 2 faithful}}                                                                                                    \\
                                          & \iff \trans(\term[1]) \not\sim_{\REL} \trans(\term[2]).
    \end{align*}
    Because $\term[1], \term[2]$ are in $\CoRPi_2$, we have that $\trans(\term[1]), \trans(\term[2])$ are in $\CoRPi_2$ (\Cref{claim: preserve k}).
    Hence, since $\trans(\term[1]), \trans(\term[2])$ contains the one variable $a$, this completes the proof.
\end{proof}
Therefore, \Cref{proposition: decidable/undecidable} holds even if $\vsig$ is an non-empty finite set (particularly if $\card \vsig = 1$):
\begin{proposition}[restatement of \Cref{proposition: decidable/undecidable}]
    Assume that $\vsig$ is a non-empty finite set.
    The \kl{equational theory} of $\CoRSigma_{n}$ (resp.\ $\CoRPi_{n}$) is decidable if $n \le 1$
    and is undecidable if $n \ge 2$.
\end{proposition}
\begin{proof}
    By \Cref{lemma: elim converse 2}.
\end{proof}

\section{Proof completion for \Cref{lemma: simpl CoR}}\label{section: lemma: simple CoR}
\subsection{Complement normal form}\label{section: complement normal form}
\begin{lemma}[equations for the complement normal form]\label{lemma: complement normal form}
    The following holds:
    \begin{align*}
        \top^{-} & \sim_{\REL} \bot & \bot^{-} & \sim_{\REL} \top & \const{I}^{-} & \sim_{\REL} \const{D} & \const{D}^{-} & \sim_{\REL} \const{I}
    \end{align*}
    \begin{align*}
        (\term[2] \cup \term[3])^{-} & \sim_{\REL} \term[2]^{-} \cap \term[3]^{-} & (\term[2] \cap \term[3])^{-} & \sim_{\REL} \term[2]^{-} \cup \term[3]^{-} & (\term[2]^{-})^{-} & \sim_{\REL} \term[2] & (\term[2]^{\pi})^{-} & \sim_{\REL} (\term[2]^{-})^{\pi}.
    \end{align*}
\end{lemma}
\begin{proof}
    Easy.
\end{proof}
\begin{remark}
    Whereas $(\term[1] \cdot \term[2])^{-} \sim_{\REL} \term[1]^{-} \dagger \term[2]^{-}$ and $(\term[1] \dagger \term[2])^{-} \sim_{\REL} \term[1]^{-} \cdot \term[2]^{-}$,
    we do not need them in the following because the complement operator does not apply to either $\cdot$ or $\dagger$ by the definition of $\CoRPi_{k}$ and $\CoRSigma_{k}$.
\end{remark}
\begin{lemma}\label{lemma: decomposition complement normal form}
    Let $\term \in \voset{(\CoRSigma_1)}{1}$ and $a \in \vsig$.
    There are $\term_0 \in \voset{\allterm^{\set{\cup, \cap, \cdot, \bot, \top, \const{I}, \const{D}} \cup \set{\pi \mid \pi \in \range{1, 2}^{\range{1, 2}}}}}{1}$ and $\term_1 \in \voset{\allterm^{\set{-}}}{1}$ such that $\term \sim_{\REL} \term_0\assign{\term_1}{a}$.
\end{lemma}
\begin{proof}
    For $\term \in \voset{(\CoRPi_0)}{1}$:
    By applying the rewriting rules induced from the equations in \Cref{lemma: complement normal form} (from left to right) as much as possible,
    we can obtain a term $\term' \in \voset{(\CoRPi_0)}{1}$ such that
    $\term \sim_{\REL} \term'$ and the complement operators ($-$) only apply to a variable.
    (Note that $\CoRPi_0 = \allterm^{\set{\cup, \cap, -, \bot, \top, \const{I}, \const{D}} \cup \set{\pi \mid \pi \in \range{1, 2}^{\range{1, 2}}}}$.)
    Then, there are $\term_0 \in \voset{\allterm^{\set{\cup, \cap, \bot, \top, \const{I}, \const{D}} \cup \set{\pi \mid \pi \in \range{1, 2}^{\range{1, 2}}}}}{1}$ and $\term_1 \in \voset{\allterm^{\set{-}}}{1}$ such that $\term' = \term_0\assign{\term_1}{a}$.
    Since $\term \sim_{\REL} \term_0\assign{\term_1}{a}$, we have obtained such $\term_0$ and $\term_1$.
    For $\term \in \voset{(\CoRSigma_1)}{1}$:
    By easy induction on $\term$ using the case for $\term \in \voset{(\CoRPi_0)}{1}$.
\end{proof}

\subsection{Projection normal form (PNF)}

\begin{lemma}[equations for the projection normal form]\label{lemma: projection normal form}
    The following holds:
    \begin{align*}
        \top^{\pi} & \sim_{\REL} \top & \bot^{\pi} & \sim_{\REL} \bot & \const{I}^{\pi} & \sim_{\REL} \begin{cases}
                                                                                                          \const{I} & (\pi(1) \neq \pi(2)) \\
                                                                                                          \top      & (\pi(1) = \pi(2))
                                                                                                      \end{cases} & \const{D}^{\pi} & \sim_{\REL} \begin{cases}
                                                                                                                                                      \const{D} & (\pi(1) \neq \pi(2)) \\
                                                                                                                                                      \bot      & (\pi(1) = \pi(2))
                                                                                                                                                  \end{cases}
    \end{align*}
    \begin{align*}
        (\term[2] \cup \term[3])^{\pi} & \sim_{\REL} \term[2]^{\pi} \cup \term[3]^{\pi} & (\term[2] \cap \term[3])^{\pi} & \sim_{\REL} \term[2]^{\pi} \cap \term[3]^{\pi} & (\term[2]^{-})^{\pi} & \sim_{\REL} (\term[2]^{\pi})^{-} & (\term[2]^{\pi})^{\pi'} & \sim_{\REL} \term[2]^{\pi \circ \pi'}
    \end{align*}
    \begin{align*}
        (\term[2] \cdot \term[3])^{\pi} & \sim_{\REL} \begin{cases}
                                                          \term[2] \cdot \term[3]                      & (\pi = \set{1 \mapsto 1, 2 \mapsto 2}) \\
                                                          \term[3]^{\smile} \cdot \term[2]^{\smile}    & (\pi = \set{1 \mapsto 2, 2 \mapsto 1}) \\
                                                          (\term[2] \cap \term[3]^{\smile}) \cdot \top & (\pi = \set{1 \mapsto 1, 2 \mapsto 1}) \\
                                                          \top \cdot (\term[2]^{\smile} \cap \term[3]) & (\pi = \set{1 \mapsto 2, 2 \mapsto 2})
                                                      \end{cases}
    \end{align*}
    \begin{align*}
        (\term[2] \dagger \term[3])^{\pi} & \sim_{\REL} \begin{cases}
                                                            \term[2] \dagger \term[3]                      & (\pi = \set{1 \mapsto 1, 2 \mapsto 2}) \\
                                                            \term[3]^{\smile} \dagger \term[2]^{\smile}    & (\pi = \set{1 \mapsto 2, 2 \mapsto 1}) \\
                                                            (\term[2] \cup \term[3]^{\smile}) \dagger \bot & (\pi = \set{1 \mapsto 1, 2 \mapsto 1}) \\
                                                            \bot \dagger (\term[2]^{\smile} \cup \term[3]) & (\pi = \set{1 \mapsto 2, 2 \mapsto 2})
                                                        \end{cases}.
    \end{align*}
\end{lemma}
\begin{proof}
    They are also easy, from the semantics.
    For example,
    for $(\term[2] \cup \term[3])^{\pi} \sim_{\REL} \term[2]^{\pi} \cup \term[3]^{\pi}$,
    \begin{align*}
        \tuple{x_1, x_2} \in \jump{(\term[2] \cup \term[3])^{\pi}}_{\model} & \iff \tuple{x_{\pi(1)}, x_{\pi(2)}} \in \jump{\term[2] \cup \term[3]}_{\model}                                                    \\
                                                                            & \iff \tuple{x_{\pi(1)}, x_{\pi(2)}} \in \jump{\term[2]}_{\model} \lor \tuple{x_{\pi(1)}, x_{\pi(2)}} \in \jump{\term[3]}_{\model}
        \\
                                                                            & \iff \tuple{x_1, x_2} \in \jump{\term[2]^{\pi}}_{\model} \lor \tuple{x_1, x_2} \in \jump{\term[3]^{\pi}}_{\model}                 \\
                                                                            & \iff \tuple{x_1, x_2} \in \jump{\term[2]^{\pi} \cup \term[3]^{\pi}}_{\model}.
    \end{align*}
    For $(\term[2]^{\pi})^{\pi'} \sim_{\REL} \term[2]^{\pi \circ \pi'}$,
    \begin{align*}
        \tuple{x_1, x_2} \in \jump{(\term[2]^{\pi})^{\pi'}}_{\model} & \iff \tuple{x_{\pi'(1)}, x_{\pi'(2)}} \in \jump{\term[2]^{\pi}}_{\model}                    \\
                                                                     & \iff \tuple{x_{\pi(\pi'(1))}, x_{\pi(\pi'(2))}} \in \jump{\term[2]}_{\model}
        \\
                                                                     & \iff \tuple{x_{(\pi \circ \pi')(1)}, x_{(\pi \circ \pi')(2)}}  \in \jump{\term[2]}_{\model} \\
                                                                     & \iff \tuple{x_1, x_2} \in \jump{\term[2]^{\pi \circ \pi'}}_{\model}.
    \end{align*}
    For $(\term[2] \cdot \term[3])^{\pi}$, we distinguish the following cases:
    Case $\pi = \set{1 \mapsto 1, 2 \mapsto 2}$: Clear.
    Case $\pi = \set{1 \mapsto 2, 2 \mapsto 1}$:
    \begin{align*}
        \tuple{x_1, x_2} \in \jump{(\term[2] \cdot \term[3])^{\pi}}_{\model} & \iff \tuple{x_{2}, x_{1}} \in \jump{\term[2] \cdot \term[3]}_{\model}                                                           \\
                                                                             & \iff \exists z,  \tuple{x_{2}, z} \in \jump{\term[2]}_{\model} \land  \tuple{z, x_{1}} \in \jump{\term[3]}_{\model}             \\
                                                                             & \iff \exists z,  \tuple{z, x_{2}} \in \jump{\term[2]^{\pi}}_{\model} \land  \tuple{x_{1}, z} \in \jump{\term[3]^{\pi}}_{\model} \\
                                                                             & \iff \tuple{x_{1}, x_{2}} \in \jump{\term[3]^{\pi} \cdot \term[2]^{\pi}}_{\model}.
    \end{align*}
    Case $\pi = \set{1 \mapsto 1, 2 \mapsto 1}$:
    \begin{align*}
        \tuple{x_1, x_2} \in \jump{(\term[2] \cdot \term[3])^{\pi}}_{\model} & \iff \tuple{x_{1}, x_{1}} \in \jump{\term[2] \cdot \term[3]}_{\model}                                                               \\
                                                                             & \iff \exists z,  \tuple{x_{1}, z} \in \jump{\term[2]}_{\model} \land  \tuple{z, x_{1}} \in \jump{\term[3]}_{\model}                 \\
                                                                             & \iff \exists z,  \tuple{x_{1}, z} \in \jump{\term[2]}_{\model} \land  \tuple{x_{1}, z} \in \jump{\term[3]^{\smile}}_{\model}        \\
                                                                             & \iff \exists z,  \tuple{x_{1}, z} \in \jump{\term[2] \cap \term[3]^{\smile}}_{\model} \land \tuple{z, x_2} \in \jump{\top}_{\model} \\
                                                                             & \iff \tuple{x_{1}, x_{2}} \in \jump{(\term[2] \cap \term[3]^{\smile}) \cdot \top}_{\model}.
    \end{align*}
    Case $\pi = \set{1 \mapsto 2, 2 \mapsto 2}$:
    As with the case of $\pi = \set{1 \mapsto 1, 2 \mapsto 1}$.
\end{proof}
\begin{lemma}[decomposition by projection normal form]\label{lemma: decomposition projection normal form}
    Let $\term \in \voset{\allterm^{\set{\cup, \cap, \cdot, \bot, \top, \const{I}, \const{D}} \cup \set{\pi \mid \pi \in \range{1, 2}^{\range{1, 2}}}}}{1}$ and $a \in \vsig$.
    Then, there are $\term_0 \in \voset{\allterm^{\set{\cup, \cap, \cdot, \bot, \top, \const{I}, \const{D}}}}{1}$ and $\term_1 \in \voset{\allterm^{\set{\pi \mid \pi \in \range{1, 2}^{\range{1, 2}}}}}{1}$ such that $\term \sim_{\REL} \term_0\assign{\term_1}{a}$.
\end{lemma}
\begin{proof}
    By applying the rewriting rules induced from the equations in \Cref{lemma: projection normal form} (from left to right) as much as possible,
    we can obtain a term $\term' \in \voset{\allterm^{\set{\cup, \cap, \cdot, \bot, \top, \const{I}, \const{D}} \cup \set{\pi \mid \pi \in \range{1, 2}^{\range{1, 2}}}}}{1}$ such that
    $\term \sim_{\REL} \term'$ and the projection operators $\pi$ only apply to a variable.
    Then,
    for the term $\term'$, there are $\term_0 \in \voset{\allterm^{\set{\cup, \cap, \bot, \top, \const{I}, \const{D}}}}{1}$ and $\term_1 \in \voset{\allterm^{\set{\pi \mid \pi \in \range{1, 2}^{\range{1, 2}}}}}{1}$ such that $\term' = \term_0\assign{\term_1}{a}$.
    Since $\term \sim_{\REL} \term_0\assign{\term_1}{a}$, we have obtained such $\term_0$ and $\term_1$.
\end{proof}

\subsection{Union normal form}
\begin{lemma}[eliminating $\bot$ and $\top$]\label{lemma: elim bot and top}
    For every $\term \in \voset{\allterm^{\set{\cup, \cap, \cdot, \bot, \top, \const{I}, \const{D}}}}{1}$,
    there is $\term_0 \in \voset{\allterm^{\set{\cup, \cap, \cdot, \const{I}, \const{D}}}}{1}$ such that $\term \sim_{\REL} \term_0$.
\end{lemma}
\begin{proof}
    By $\bot \sim_{\REL} \const{I} \cap \const{D}$ and $\top \sim_{\REL} \const{I} \cup \const{D}$.
    let $\term_0$ be the term $\term$ in which $\bot$ has been replaced with $\const{I} \cap \const{D}$ and $\top$ has been replaced with $\const{I} \cup \const{D}$.
    Then, $\term \sim_{\REL} \term_0$ and $\term_0 \in \voset{\allterm^{\set{\cup, \cdot, \cap, \const{I}, \const{D}}}}{1}$.
\end{proof}

\begin{lemma}[equations for the union normal form]\label{lemma: union normal form}
    The following holds:
    \begin{align*}
        (\term[2]_1 \cup \term[2]_2) \cap \term[3]  & \sim_{\REL} (\term[2]_1 \cap \term[3]) \cup (\term[2]_2 \cap \term[3])   & \term[3] \cap (\term[2]_1 \cup \term[2]_2)  & \sim_{\REL} (\term[3] \cap \term[2]_1) \cup (\term[3] \cap \term[2]_2)    \\
        (\term[2]_1 \cup \term[2]_2) \cdot \term[3] & \sim_{\REL} (\term[2]_1 \cdot \term[3]) \cup (\term[2]_2 \cdot \term[3]) & \term[3] \cdot (\term[2]_1 \cup \term[2]_2) & \sim_{\REL} (\term[3] \cdot \term[2]_1) \cup (\term[3] \cdot \term[2]_2).
    \end{align*}
\end{lemma}
\begin{proof}
    Easy.
\end{proof}

\begin{lemma}[decomposition by union normal form]\label{lemma: decomposition union normal form}
    For every $\term \in \voset{\allterm^{\set{\cup, \cap, \cdot, \const{I}, \const{D}}}}{1}$,
    there are $n \ge 1$ and $\term_1, \dots, \term_n \in \voset{\allterm^{\set{\cap, \cdot, \const{I}, \const{D}}}}{1}$ such that $\term \sim_{\REL} \term_1 \cup \dots \cup \term_n$.
\end{lemma}
\begin{proof}
    By straightforward induction on $\term$ using the equations in \Cref{lemma: union normal form}.
\end{proof}

\subsection{Proof of \Cref{lemma: simpl CoR}}

\begin{lemma}[restatement of \Cref{lemma: simpl CoR}]
    If $\voset{\allterm^{\set{\cap, \cdot, \const{I}, \const{D}}}}{1}\quoset{\sim_{\REL}}$ is finite,
    $\voset{(\CoRSigma_1)}{1}\quoset{\sim_{\REL}}$ is finite.
\end{lemma}
\begin{proof}
    We have:
    \begin{align*}
         & \mbox{$\voset{\allterm^{\set{\cap, \cdot, \const{I}, \const{D}}}}{1}\quoset{\sim_{\REL}}$ is finite }  \tag*{($\because$ the assumption)}                                                                                                                                                                                                                                          \\
         & \mbox{$\Longrightarrow$ $\voset{\allterm^{\set{\cup, \cap, \cdot, \const{I}, \const{D}}}}{1}\quoset{\sim_{\REL}}$ is finite }  \tag*{($\because$ \Cref{lemma: decomposition union normal form})}                                                                                                                                                                                   \\
         & \mbox{$\Longrightarrow$ $\voset{\allterm^{\set{\cup, \cap, \cdot, \bot, \top, \const{I}, \const{D}}}}{1}\quoset{\sim_{\REL}}$ is finite} \tag*{($\because$ \Cref{lemma: elim bot and top})}                                                                                                                                                                                        \\
         & \mbox{$\Longrightarrow$ $\voset{\allterm^{\set{\cup, \cap, \cdot, \bot, \top, \const{I}, \const{D}} \cup \set{\pi \mid \pi \in \range{1, 2}^{\range{1, 2}}}}}{1}\quoset{\sim_{\REL}}$ is finite} \tag*{(\Cref{lemma: decomposition projection normal form} with the finiteness of $\voset{\allterm^{\set{\pi \mid \pi \in \range{1, 2}^{\range{1, 2}}}}}{1}\quoset{\sim_{\REL}}$)} \\
         & \mbox{$\Longrightarrow$ $\voset{(\CoRSigma_1)}{1}\quoset{\sim_{\REL}}$ is finite.} \tag*{(\Cref{lemma: decomposition complement normal form} with the finiteness of $\voset{\allterm^{\set{-}}}{1}\quoset{\sim_{\REL}}$)}
    \end{align*}
\end{proof}
\section{Proof Completion of \Cref{lemma: soundness}}\label{section: lemma: soundness}
\begin{remark}\label{remark: SMT}
    Using the standard encoding to formulas of first-order logic \cite{Tarski1941}, we can automatically check the validity of these equations.
    See \cite{nakamuraRepositoryFiniteVariableOccurrence2023} for the TPTP files---they are at least checked by Z3 (Z3tptp 4.8.11.0) \cite{demouraZ3EfficientSMT2008}, Vampire 4.7 (linked with Z3 4.8.13.0) \cite{kovacsFirstOrderTheoremProving2013}, and CVC5 1.0.3 \cite{barbosaCvc5VersatileIndustrialStrength2022}.
    Our encoding is based on the encoding into first-order formulas \cite{Tarski1941} (see also \cite{nakamuraExpressivePowerSuccinctness2020,nakamuraExpressivePowerSuccinctness2022}).
    Note that there is another earlier presented TPTP-encoding for the calculus of relations, by H{\"o}fner and Struth \cite{hofnerAutomatingCalculusRelations2008}, which is based on axioms of relation algebras.
\end{remark}
Apart from the automated checking above, in the following, we present an explicit proof for each equation (w.r.t.\ $\sim_{\REL_{\ge 5}}$).

\noindent
For $\iI = \iI\iI$:
\begin{align*}
    a \cap \const{I} & = a \cap (\const{I} \cap \const{I}) \tag{$\const{I} \sim_{\REL} \const{I} \cap \const{I}$} \\
                     & = (a \cap \const{I}) \cap \const{I}. \tag{associativity law}
\end{align*}

\noindent
For $\iD  = \iD\iD$:
\begin{align*}
    a \cap \const{D} & = a \cap (\const{D} \cap \const{D}) \tag{$\const{D} \sim_{\REL} \const{D} \cap \const{D}$} \\
                     & = (a \cap \const{D}) \cap \const{D}.  \tag{associativity law}
\end{align*}

\noindent
For $\iI        = \iI\re$:
\begin{align*}
    a \cap \const{I} & = a^{\smile} \cap \const{I}.  \tag{$\term \cap \const{I} \sim_{\REL} \term^{\smile} \cap \const{I}$}
\end{align*}

\noindent
For $\iI   = \re\iI$:
\begin{align*}
    a \cap \const{I} & = a^{\smile} \cap \const{I}                                                           \tag{$\term \cap \const{I} \sim_{\REL} \term^{\smile} \cap \const{I}$} \\
                     & = (a \cap \const{I})^{\smile}. \tag{PNF; \Cref{lemma: projection normal form}}
\end{align*}

\noindent
For $\iI\iD      = \iD\iI $:
\begin{align*}
    (a \cap \const{D}) \cap \const{I} & = a \cap (\const{D} \cap \const{I})                                                                     \tag{associativity law} \\
                                      & = a \cap (\const{I} \cap \const{D}) \tag{commutativity law}                                                                     \\
                                      & = (a \cap \const{I}) \cap \const{D}. \tag{associativity law}
\end{align*}

\noindent
For $\iD\re  = \re\iD $:
\begin{align*}
    a^{\smile} \cap \const{D} & = (a \cap \const{D})^{\smile}. \tag{PNF}
\end{align*}

\noindent
For $\empword = \re\re$:
\begin{align*}
    a & = a^{\smile\smile}. \tag{PNF}
\end{align*}

\noindent
For $\iD\iI     = \cD\iI\iD$:
\begin{align*}
    (a \cap \const{I}) \cap \const{D} & = a \cap (\const{I} \cap \const{D})  \tag{associativity law}                                             \\
                                      & = a \cap \bot \tag{$\bot \sim_{\REL} \const{I} \cap \const{D}$}                                          \\
                                      & = \bot \tag{$\bot \sim_{\REL} \term \cap \bot$}                                                          \\
                                      & = \bot \cdot \const{D}  \tag{$\bot \sim_{\REL} \bot \cdot \term$}                                        \\
                                      & = (a \cap \bot) \cdot \const{D}  \tag{$\bot \sim_{\REL} \term \cap \bot$}                                \\
                                      & = (a \cap (\const{D} \cap \const{I})) \cdot \const{D}  \tag{$\bot \sim_{\REL} \const{D} \cap \const{I}$} \\
                                      & = ((a \cap \const{D}) \cap \const{I}) \cdot \const{D}.  \tag{associativity law}
\end{align*}

\noindent
For $\iD\iI      = \iI\cD\iI $:
\begin{align*}
    (a \cap \const{I}) \cap \const{D} & = a \cap (\const{I} \cap \const{D})  \tag{associativity law}           \\
                                      & = a \cap \bot \tag{$\bot \sim_{\REL} \const{I} \cap \const{D}$}        \\
                                      & = \bot \tag{$\bot \sim_{\REL} \term \cap \bot$}                        \\
    \label{formula: star 1}           & = ((a \cap \const{I}) \cdot \const{D}) \cap \const{I}.  \tag{$\star$1}
\end{align*}
Here, for \Cref{formula: star 1}, we consider the translated first-order formula:
\[\const{False} \leftrightarrow ((\exists y_1, (a(x_0, y_1) \land x_0 = y_1) \land y_1 \neq y_0) \land x_0 = y_0).\]
This formula is valid over $\REL$ because $x_0 = y_1 \land y_1 \neq y_0 \land x_0 = y_0$ is unsatisfiable.

\noindent
For $\iI\cD      = \iI\cD\iD$:
\begin{align*}
    \label{formula: star 2} (a \cdot \const{D}) \cap \const{I} & = ((a \cap \const{D}) \cdot \const{D}) \cap \const{I} \tag{$\star$2}
\end{align*}
Here, for \Cref{formula: star 2}, we consider the translated first-order formula:
\[((\exists y_1, a(x_0, y_1) \land y_1 \neq y_0) \land x_0 = y_0) \leftrightarrow ((\exists y_1, (a(x_0, y_1) \land x_0 \neq y_1) \land y_1 \neq y_0) \land x_0 = y_0).\]
This formula is valid over $\REL$ because $(y_1 \neq y_0 \land x_0 = y_0) \leftrightarrow (x_0 \neq y_1 \land y_1 \neq y_0 \land x_0 = y_0)$ always holds.

\noindent
For $\cD\iI      = \iD\cD\iI$:
\begin{align*}
    \label{formula: star 3} (a \cap \const{I}) \cdot \const{D} & = ((a \cap \const{I}) \cdot \const{D}) \cap \const{D} \tag{$\star$3}
\end{align*}
Here, for \Cref{formula: star 3}, we consider the translated first-order formula:
\[((\exists y_1, a(x_0, y_1) \land x_0 = y_1) \land y_1 \neq y_0) \leftrightarrow ((\exists y_1, (a(x_0, y_1) \land x_0 = y_1) \land y_1 \neq y_0) \land x_0 \neq y_0).\]
This formula is valid over $\REL$ because $(y_1 \neq y_0 \land x_0 = y_0) \leftrightarrow (x_0 \neq y_1 \land y_1 \neq y_0 \land x_0 = y_0)$ always holds.

\noindent
For $\cD\cD      = \cD\iD\cD$:
\begin{align*}
    (a \cdot \const{D}) \cdot \const{D} & = a \cdot (\const{D} \cdot \const{D})  \tag{associativity law}          \\
                                        & = a \cdot \top \tag{$\top \sim_{\REL} \const{D} \cdot \const{D}$}       \\
    \label{formula: star 4}             & = ((a \cdot \const{D}) \cap \const{D}) \cdot \const{D}.  \tag{$\star$4}
\end{align*}
Here, for \Cref{formula: star 1}, we consider the translated first-order formula:
\[(\exists y_1, a(x_0, y_1)) \leftrightarrow (\exists y_2, ((\exists y_1, a(x_0, y_1) \land y_1 \neq y_2) \land x_0 \neq y_2) \land y_2 \neq y_0).\]
For the right-hand side formula, over $\REL_{\ge 4}$, we have:
\begin{align*}
     & (\exists y_2, ((\exists y_1, a(x_0, y_1) \land y_1 \neq y_2) \land x_0 \neq y_2) \land y_2 \neq y_0)                                                                   \\
     & \leftrightarrow (\exists y_2 \ y_1, a(x_0, y_1) \land y_1 \neq y_2 \land x_0 \neq y_2 \land y_2 \neq y_0) \tag{prenex normal form (prenex)}                            \\
     & \leftrightarrow (\exists y_1, a(x_0, y_1) \land (\exists y_2, y_2 \neq y_1 \land y_2 \neq x_0 \land y_2 \neq y_0)) \tag{prenex}                                        \\
     & \leftrightarrow (\exists y_1, a(x_0, y_1)). \tag{$\const{True} \leftrightarrow (\exists y_2, y_2 \neq y_1 \land y_2 \neq x_0 \land y_2 \neq y_0)$ over $\REL_{\ge 4}$}
\end{align*}
Hence, the translated formula for \Cref{formula: star 4} is valid over $\REL_{\ge 4}$.

\noindent
For $\cD\cD      = \cD\cD\cD$:
\begin{align*}
    (a \cdot \const{D}) \cdot \const{D} & = a \cdot \top   \tag{$\cdot_{\top} = \cD\cD$}                                          \\
                                        & = a \cdot (\top \cdot \const{D}) \tag{$\top \sim_{\REL_{\ge 2}} \top \cdot \const{D}$}  \\
                                        & = (a \cdot \top) \cdot \const{D} \tag{associativity law}                                \\
                                        & = ((a \cdot \const{D}) \cdot \const{D}) \cdot \const{D}.  \tag{$\cdot_{\top} = \cD\cD$}
\end{align*}
Here, ($\cdot_{\top} = \cD\cD$) means the following:
\begin{align*}
    \term \cdot \top & =  \term \cdot (\const{D} \cdot \const{D})  \tag{$\top \sim_{\REL_{\ge 3}} \const{D} \cdot \const{D}$} \\
                     & =  (\term \cdot \const{D}) \cdot \const{D}\tag{associativity law}
\end{align*}

\noindent
For $\cD\cD\iD     = \cD\cD\iI\cD$:
\begin{align*}
    ((a \cap \const{D}) \cdot \const{D}) \cdot \const{D} & = (a \cap \const{D}) \cdot \top \tag{$\cdot_{\top} = \cD\cD$}                                             \\
    \label{formula: star 5}                              & = ((a \cdot \const{D}) \cap \const{I}) \cdot \top  \tag{$\star$5}                                         \\
                                                         & = (((a \cdot \const{D}) \cap \const{I}) \cdot \const{D}) \cdot \const{D}.   \tag{$\cdot_{\top} = \cD\cD$}
\end{align*}
Here, for \Cref{formula: star 5}, we consider the translated first-order formula:
\[(\exists y_1, a(x_0, y_1) \land x_0 \neq y_1) \leftrightarrow (\exists y_2, (\exists y_1, a(x_0, y_1) \land y_1 \neq y_2) \land x_0 = y_2).\]
This formula is valid over $\REL$, which can be shown by using ${\fml}[x_0/y_2] \leftrightarrow (\exists y_2, \fml \land x_0 = y_2)$.

\noindent
For $\iD\cD\cD     = \cD\iI\cD\cD$:
\begin{align*}
    ((a \cdot \const{D}) \cdot \const{D}) \cap \const{D} & = (a \cdot \top) \cap \const{D} \tag{$\cdot_{\top} = \cD\cD$}                                           \\
    \label{formula: star 6}                              & = ((a \cdot \top) \cap \const{I}) \cdot \const{D}  \tag{$\star$6}                                       \\
                                                         & = ((a \cdot \const{D} \cdot \const{D}) \cap \const{I}) \cdot \const{D}.   \tag{$\cdot_{\top} = \cD\cD$}
\end{align*}
Here, for \Cref{formula: star 6}, we consider the translated first-order formula:
\[((\exists y_1, a(x_0, y_1)) \land x_0 \neq y_0) \leftrightarrow (\exists y_1, ((\exists y_2, a(x_0, y_2)) \land x_0 = y_1) \land y_1 \neq y_0).\]
This formula is valid over $\REL$, which can be shown by using ${\fml}[y_0/y_1] \leftrightarrow (\exists y_1, \fml \land y_1 = y_0)$.

\noindent
For $\cD\re\cD\re    = \re\cD\re\cD $:
\begin{align*}
    (a^{\smile} \cdot \const{D})^{\smile} \cdot \const{D} & = (\const{D} \cdot a) \cdot \const{D}        \tag{PNF}                \\
                                                          & = \const{D} \cdot (a \cdot \const{D})  \tag{associativity law}        \\
                                                          & =  ((a \cdot \const{D})^{\smile} \cdot \const{D})^{\smile}. \tag{PNF}
\end{align*}

\noindent
For $\re\cD\iI     = \iD\re\cD\iI$:
\begin{align*}
    ((a \cap \const{I}) \cdot \const{D})^{\smile} & = \const{D} \cdot (a^{\smile} \cap \const{I})                  \tag{PNF}      \\
    \label{formula: star 7}                       & = (\const{D} \cdot (a^{\smile} \cap \const{I})) \cap \const{D} \tag{$\star$7} \\
                                                  & = ((a \cap \const{I}) \cdot \const{D})^{\smile} \cap \const{D}. \tag{PNF}
\end{align*}
Here, for \Cref{formula: star 7}, we consider the translated first-order formula:
\[(\exists x_1, x_0 \neq x_1 \land (a(y_0, x_1) \land x_1 = y_0)) \leftrightarrow (\exists x_1, x_0 \neq x_1 \land (a(y_0, x_1) \land x_1 = y_0)) \land x_0 \neq y_0.\]
This formula is valid over $\REL$ because $(x_0 \neq x_1 \land x_1 = y_0) \leftrightarrow (x_0 \neq x_1 \land x_1 = y_0 \land x_0 \neq y_0)$ always holds.

\noindent
For $\re\cD\iI\cD   = \iD\re\cD\cD\iD$:
\begin{align*}
    (((a \cdot \const{D}) \cap \const{I}) \cdot \const{D})^{\smile} & = \const{D} \cdot ((\const{D} \cdot a^{\smile}) \cap \const{I})                \tag{PNF}                                                                   \\
    \label{formula: star 8}                                         & = (\top \cdot (a^{\smile} \cap \const{D})) \cap \const{D}   \tag{$\star$8}                                                                                 \\
                                                                    & = ((\const{D} \cdot \const{D}) \cdot (a^{\smile} \cap \const{D})) \cap \const{D}                \tag{$\top \sim_{\REL_{\ge 3}} \const{D} \cdot \const{D}$} \\
                                                                    & = (\const{D} \cdot (\const{D} \cdot (a^{\smile} \cap \const{D}))) \cap \const{D} \tag{associativity law}                                                   \\
                                                                    & = (((a \cap \const{D}) \cdot \const{D}) \cdot \const{D})^{\smile} \cap \const{D}. \tag{PNF}
\end{align*}
Here, for \Cref{formula: star 8}, we consider the translated first-order formula:
\begin{align*}
     & (\exists x_2, x_0 \neq x_2 \land ((\exists x_1, x_2 \neq x_1 \land a(y_0, x_1)) \land x_2 = y_0)) \\
     & \leftrightarrow ((\exists x_1, a(y_0, x_1) \land x_1 \neq y_0) \land x_0 \neq y_0).
\end{align*}
This formula is valid over $\REL$, which can be shown by using ${\fml}[y_0/x_2] \leftrightarrow (\exists x_2, \fml \land x_2 = y_0)$.

\noindent
For $\cD\re\cD\re = \re\cD\re\cD$:
\begin{align*}
    (a^{\smile} \cdot \const{D})^{\smile} \cdot \const{D} & = (\const{D} \cdot a) \cdot \const{D} \tag{PNF}                     \\
                                                          & = \const{D} \cdot (a \cdot \const{D}) \tag{associativity law}       \\
                                                          & = ((a \cdot \const{D})^{\smile} \cdot \const{D})^{\smile} \tag{PNF}
\end{align*}

\noindent
For $\re\cD\iI\cD = \iD\re\cD\cD\iD$:
\begin{align*}
    (((a \cdot \const{D}) \cap \const{I}) \cdot \const{D})^{\smile} & = \const{D} \cdot ((\const{D} \cdot a^{\smile}) \cap \const{I}) \tag{PNF}                                         \\
    \label{formula: star 9}                                         & = ((\top \cdot (a^{\smile} \cap \const{D})) \cap \const{D}) \tag{$\star$9}                                        \\
                                                                    & = (((a \cap \const{D}) \cdot \top)^{\smile} \cap \const{D}) \tag{PNF}                                             \\
                                                                    & = ((((a \cap \const{D}) \cdot \const{D}) \cdot \const{D})^{\smile} \cap \const{D})  \tag{$\cdot_{\top} = \cD\cD$}
\end{align*}
Here, for \Cref{formula: star 9}, we consider the translated first-order formula:
\begin{align*}
     & (\exists x_1, x_0 \neq x_1 \land ((\exists x_2, x_1 \neq x_2 \land a(y_0, x_2)) \land x_1 = y_0)) \\
     & \leftrightarrow ((\exists x_2, a(y_0, x_2) \land x_2 \neq y_0) \land x_0 \neq y_0).
\end{align*}
This formula is valid over $\REL$, which can be shown by using ${\fml}[y_0/x_1] \leftrightarrow (\exists x_1, \fml \land x_1 = y_0)$.

\noindent
For $ \cD\re\cD\cD\re = \re\cD\cD\re\cD$:
\begin{align*}
     & (((a \cdot \const{D})^{\smile} \cdot \const{D}) \cdot \const{D})^{\smile}                                \\
     & = ((a \cdot \const{D})^{\smile} \cdot \top)^{\smile} \tag{$\cdot_{\top} = \cD\cD$}                       \\
     & = \top \cdot (a \cdot \const{D})  \tag{PNF}                                                              \\
     & = (\top \cdot a) \cdot \const{D}   \tag{associativity law}                                               \\
     & = (a^{\smile} \cdot \top)^{\smile} \cdot \const{D}                                  \tag{PNF}            \\
     & = ((a^{\smile} \cdot \const{D}) \cdot \const{D})^{\smile} \cdot \const{D}. \tag{$\cdot_{\top} = \cD\cD$}
\end{align*}

\noindent
For $\cD\iD\re\cD\iD\re\cD\iD\re\cD\iD\re\cD\iD = \re\cD\iD\re\cD\iD\re\cD\iD\re\cD\iD\re\cD\iD\re$:
\begin{align*}
                             & (((((((((a \cap \const{D}) \cdot \const{D})^{\smile} \cap \const{D}) \cdot \const{D})^{\smile} \cap \const{D}) \cdot \const{D})^{\smile}
    \cap \const{D}) \cdot \const{D})^{\smile} \cap \const{D}) \cdot \const{D}                                                                                                                      \\
                             & = ((\const{D} \cdot (((((((a \cap \const{D}) \cdot \const{D})^{\smile} \cap \const{D}) \cdot \const{D})^{\smile} \cap \const{D}) \cdot \const{D})
    \cap \const{D})) \cap \const{D}) \cdot \const{D}  \tag{PNF}                                                                                                                                    \\
                             & = ((\const{D} \cdot ((((\const{D} \cdot (((a \cap \const{D}) \cdot \const{D}) \cap \const{D})) \cap \const{D}) \cdot \const{D})
    \cap \const{D})) \cap \const{D}) \cdot \const{D}  \tag{PNF}                                                                                                                                    \\
    \label{formula: star 10} & = \const{D} \cdot ((((\const{D} \cdot ((((\const{D} \cdot (a \cap \const{D})) \cap \const{D}) \cdot \const{D}) \cap \const{D}))
    \cap \const{D}) \cdot \const{D}) \cap \const{D})   \tag{$\star$10}                                                                                                                             \\
                             & = \const{D} \cdot ((((\const{D} \cdot (((((a^{\smile} \cap \const{D}) \cdot \const{D})^{\smile} \cap \const{D}) \cdot \const{D}) \cap \const{D}))
    \cap \const{D}) \cdot \const{D}) \cap \const{D})   \tag{PNF}                                                                                                                                   \\
                             & = \const{D} \cdot (((((((((a^{\smile} \cap \const{D}) \cdot \const{D})^{\smile} \cap \const{D}) \cdot \const{D})^{\smile} \cap \const{D}) \cdot \const{D})^{\smile}
    \cap \const{D}) \cdot \const{D}) \cap \const{D})   \tag{PNF}                                                                                                                                   \\
                             & = ((((((((((a^{\smile} \cap \const{D}) \cdot \const{D})^{\smile} \cap \const{D}) \cdot \const{D})^{\smile} \cap \const{D}) \cdot \const{D})^{\smile}
    \cap \const{D}) \cdot \const{D})^{\smile} \cap \const{D}) \cdot \const{D})^{\smile}.   \tag{PNF}
\end{align*}
Here, for \Cref{formula: star 10}, we consider the translated first-order formula:
\begin{align*}
     & (\exists y_1, ( (\exists x_1, x_0 \neq x_1 \land ((\exists y_2, (                                                               \\
     & \quad (\exists x_2, x_1 \neq x_2 \land ((\exists y_3, (a(x_2, y_3) \land x_2 \neq y_3) \land y_3 \neq y_2) \land x_2 \neq y_2)) \\
     & \quad \land x_1 \neq y_2) \land y_2 \neq y_1) \land x_1 \neq y_1)) \land x_0 \neq y_1) \land y_1 \neq y_0)                      \\
     & \leftrightarrow                                                                                                                 \\
     & (\exists x_1, (\exists y_1, ((\exists x_2, x_1 \neq x_2 \land ((\exists y_2, (                                                  \\
     & \quad (\exists x_3, x_2 \neq x_3 \land (a(x_3, y_2) \land x_3 \neq y_2))                                                        \\
     & \quad \land x_2 \neq y_2) \land y_2 \neq y_1) \land x_2 \neq y_1)) \land x_1 \neq y_1) \land y_1 \neq y_0) \land x_1 \neq y_0).
\end{align*}
By taking the prenex normal form on each side, this formula is equivalent to the following:
\begin{align*}
    \left(\begin{aligned}
               & \exists x_1 \ x_2 \ y_1 \ y_2 \ y_3, a(x_2, y_3)                \\
               & \quad \land x_0 \neq x_1 \land x_1 \neq x_2  \land x_2 \neq y_3 \\
               & \quad \land y_3 \neq y_2 \land y_2 \neq y_1 \land y_1 \neq y_0  \\
               & \quad \land x_0 \neq y_1 \land x_1 \neq y_1                     \\
               & \quad \land x_1 \neq y_2 \land x_2 \neq y_2
          \end{aligned}\right)
     & \leftrightarrow
    \left(\begin{aligned}
               & \exists x_1 \ x_2 \ x_3 \ y_1 \ y_2, a(x_3, y_2)               \\
               & \quad \land x_0 \neq x_1 \land x_1 \neq x_2 \land x_2 \neq x_3 \\
               & \quad \land x_3 \neq y_2 \land y_2 \neq y_1 \land y_1 \neq y_0 \\
               & \quad \land x_1 \neq y_0 \land x_1 \neq y_1                    \\
               & \quad \land x_2 \neq y_1 \land x_2 \neq y_2
          \end{aligned}\right).
\end{align*}
For the left-hand side formula, over $\REL_{\ge 5}$, we have:
\begin{align*}
    \left(\begin{aligned}
                   & \exists x_1 \ x_2 \ y_1 \ y_2 \ y_3, a(x_2, y_3)                \\
                   & \quad \land x_0 \neq x_1 \land x_1 \neq x_2  \land x_2 \neq y_3 \\
                   & \quad \land y_3 \neq y_2 \land y_2 \neq y_1 \land y_1 \neq y_0  \\
                   & \quad \land x_0 \neq y_1 \land x_1 \neq y_1                     \\
                   & \quad \land x_1 \neq y_2 \land x_2 \neq y_2
              \end{aligned}\right)
     & \leftrightarrow
    \left(\begin{aligned}
                   & \exists \ x_2 \ y_3, a(x_2, y_3)  \land x_2 \neq y_3                             \\
                   & \quad \land \exists y_2, y_2 \neq y_3 \land y_2 \neq x_2                         \\
                   & \quad\quad \land \exists x_1, x_1 \neq x_0 \land x_1 \neq x_2 \land x_1 \neq y_2 \\
                   & \quad\quad\quad \land \exists y_1,  y_1 \neq y_2                                 \\
                   & \hspace{5em}\land y_1 \neq y_0 \land y_1 \neq x_0 \land y_1 \neq x_1             \\
              \end{aligned}\right) \tag{prenex}                                                                                       \\
     & \leftrightarrow
    \left(\begin{aligned}
                   & \exists \ x_2 \ y_3, a(x_2, y_3)  \land x_2 \neq y_3                             \\
                   & \quad \land \exists y_2, y_2 \neq y_3 \land y_2 \neq x_2                         \\
                   & \quad\quad \land \exists x_1, x_1 \neq x_0 \land x_1 \neq x_2 \land x_1 \neq y_2
              \end{aligned}\right) \tag{$\const{True} \leftrightarrow \exists y_1,  y_1 \neq y_2 \land y_1 \neq y_0 \land y_1 \neq x_0 \land y_1 \neq x_1$ over $\REL_{\ge 5}$}         \\
     & \leftrightarrow
    \left(\begin{aligned}
                   & \exists \ x_2 \ y_3, a(x_2, y_3)  \land x_2 \neq y_3     \\
                   & \quad \land \exists y_2, y_2 \neq y_3 \land y_2 \neq x_2
              \end{aligned}\right) \tag{$\const{True} \leftrightarrow \exists x_1, x_1 \neq x_0 \land x_1 \neq x_2 \land x_1 \neq y_2$ over $\REL_{\ge 4}$}                             \\
     & \leftrightarrow (\exists \ x_2 \ y_3, a(x_2, y_3)  \land x_2 \neq y_3) \tag{$\const{True} \leftrightarrow \exists y_2, y_2 \neq y_3 \land y_2 \neq x_2$ over $\REL_{\ge 3}$} \\
     & \leftrightarrow (\exists \ x \ y, a(x, y)  \land x \neq y).
\end{align*}
By the same argument, for the right-hand side formula, over $\REL_{\ge 5}$, we have:
\begin{align*}
    \left(\begin{aligned}
               & \exists x_1 \ x_2 \ x_3 \ y_1 \ y_2, a(x_3, y_2)               \\
               & \quad \land x_0 \neq x_1 \land x_1 \neq x_2 \land x_2 \neq x_3 \\
               & \quad \land x_3 \neq y_2 \land y_2 \neq y_1 \land y_1 \neq y_0 \\
               & \quad \land x_1 \neq y_0 \land x_1 \neq y_1                    \\
               & \quad \land x_2 \neq y_1 \land x_2 \neq y_2
          \end{aligned}\right) & \leftrightarrow (\exists \ x \ y, a(x, y)  \land x \neq y).
\end{align*}
Combining them, the translated formula from \Cref{formula: star 9} is valid over $\REL_{\ge 5}$.

\noindent
For $\cD\iI\cD\re\cD\iI\cD\re\cD = \re\cD\iI\cD\re\cD\iI\cD\re\cD$:
\begin{align*}
                             & (((((((a \cdot \const{D})^{\smile} \cdot \const{D}) \cap \const{I}) \cdot \const{D})^{\smile} \cdot \const{D}) \cap \const{I}) \cdot \const{D})                                                   \\
                             & = ((((\const{D} \cdot ((\const{D} \cdot (a \cdot \const{D})) \cap \const{I})) \cdot \const{D}) \cap \const{I}) \cdot \const{D})                                                    \tag{PNF}      \\
    \label{formula: star 11} & = (\const{D} \cdot ((\const{D} \cdot ((((\const{D} \cdot a^{\smile}) \cdot \const{D}) \cap \const{I}) \cdot \const{D})) \cap \const{I}))                             \tag{$\star$11}              \\
                             & = (((((((a \cdot \const{D})^{\smile} \cdot \const{D}) \cap \const{I}) \cdot \const{D})^{\smile} \cdot \const{D}) \cap \const{I}) \cdot \const{D})^{\smile}.                             \tag{PNF}
\end{align*}
Here, for \Cref{formula: star 11}, we consider the translated first-order formula:
\begin{align*}
     & (\exists y_1, ( (\exists y_2, (\exists x_1, x_0 \neq x_1 \land ((\exists x_2,                                                                   \\
     & x_1 \neq x_2 \land (\exists y_3, a(x_2, y_3) \land y_3 \neq y_2)) \land x_1 = y_2)) \land y_2 \neq y_1) \land x_0 = y_1) \land y_1 \neq y_0)    \\
     & \leftrightarrow                                                                                                                                 \\
     & (\exists x_1, x_0 \neq x_1 \land ((\exists x_2, x_1 \neq x_2 \land (\exists y_1,                                                                \\
     & \quad ((\exists y_2, (\exists x_3, x_2 \neq x_3 \land a(y_2, x_3)) \land y_2 \neq y_1) \land x_2 = y_1) \land y_1 \neq y_0)) \land x_1 = y_0)).
\end{align*}
By taking the prenex normal form on each side, this formula is equivalent to the following:
\begin{align*}
    \left(\begin{aligned}
               & \exists x_1 \ x_2 \ y_1 \ y_2 \ y_3, a(x_2, y_3)               \\
               & \quad \land x_0 \neq x_1 \land x_1 \neq x_2                    \\
               & \quad \land y_0 \neq y_1 \land y_1 \neq y_2 \land y_2 \neq y_3 \\
               & \quad \land x_0 = y_1 \land x_1 = y_2
          \end{aligned}\right)
     & \leftrightarrow
    \left(\begin{aligned}
               & \exists x_1 \ x_2 \ x_3 \ y_1 \ y_2, a(y_2, x_3)               \\
               & \quad \land x_0 \neq x_1 \land x_1 \neq x_2 \land x_2 \neq x_3 \\
               & \quad \land y_0 \neq y_1 \land y_1 \neq y_2                    \\
               & \quad \land x_1 = y_0 \land x_2 = y_1
          \end{aligned}\right).
\end{align*}
For the left-hand side formula, we have:
\begin{align*}
    \left(\begin{aligned}
                   & \exists x_1 \ x_2 \ y_1 \ y_2 \ y_3, a(x_2, y_3)               \\
                   & \quad \land x_0 \neq x_1 \land x_1 \neq x_2                    \\
                   & \quad \land y_0 \neq y_1 \land y_1 \neq y_2 \land y_2 \neq y_3 \\
                   & \quad \land x_0 = y_1 \land x_1 = y_2
              \end{aligned}\right)
     & \leftrightarrow
    \left(\begin{aligned}
                   & \exists x_1 \ x_2 \ y_3, a(x_2, y_3)                           \\
                   & \quad \land x_0 \neq x_1 \land x_1 \neq x_2                    \\
                   & \quad \land y_0 \neq x_0 \land x_0 \neq x_1 \land x_1 \neq y_3
              \end{aligned}\right) \tag{${\fml}[z'/z] \leftrightarrow \exists z, \fml \land z = z'$}                                                                                                      \\
     & \leftrightarrow (\exists x_2 \ y_3, a(x_2, y_3)  \land x_0 \neq y_0) \tag{$\const{True} \leftrightarrow \exists x_1,  x_1 \neq x_0 \land x_1 \neq x_2 \land x_1 \neq y_3$ over $\REL_{\ge 5}$} \\
     & \leftrightarrow (\exists x \ y, a(x, y) \land x_0 \neq y_0). \tag{By renaming}
\end{align*}
For the right-hand side formula, we have:
\begin{align*}
    \left(\begin{aligned}
                   & \exists x_1 \ x_2 \ x_3 \ y_1 \ y_2, a(y_2, x_3)               \\
                   & \quad \land x_0 \neq x_1 \land x_1 \neq x_2 \land x_2 \neq x_3 \\
                   & \quad \land y_0 \neq y_1 \land y_1 \neq y_2                    \\
                   & \quad \land x_2 = y_1 \land x_1 = y_0                          \\
              \end{aligned}\right)
     & \leftrightarrow
    \left(\begin{aligned}
                   & \exists x_3 \ y_1 \ y_2, a(y_2, x_3)                           \\
                   & \quad \land x_0 \neq y_0 \land y_0 \neq y_1 \land y_1 \neq x_3 \\
                   & \quad \land y_0 \neq y_1 \land y_1 \neq y_2                    \\
              \end{aligned}\right) \tag{${\fml}[z'/z] \leftrightarrow \exists z, \fml \land z = z'$}                                                                                                                       \\
     & \leftrightarrow (\exists x_3 \ y_2, a(y_2, x_3) \land x_0 \neq y_0) \tag{$\const{True} \leftrightarrow \exists y_1, y_1 \neq y_0 \land y_1 \neq x_3 \land y_1 \neq y_0 \land y_1 \neq y_2$ over $\REL_{\ge 5}$} \\
     & \leftrightarrow (\exists x \ y, a(x, y) \land x_0 \neq y_0). \tag{By renaming}
\end{align*}

Hence, we have proved all the equations in \Cref{figure: equations}.

\section{The minimal DFA and SMT-LIB2 file for \Cref{lemma: finite lang}}\label{appendix: SMT-LIB}
(The files in this section can also be seen in \cite{nakamuraRepositoryFiniteVariableOccurrence2023}.)
\Cref{universal} is the SMT-LIB2 file for showing:
there is no \kl{word} $w$ of length $\wlen{w} \ge 29$ such that $w \in \extpvsigzero^* \setminus (\bigcup_{i \in \range{1, 21}} \extpvsigzero^* \word[2]_i \extpvsigzero^*)$.
\Cref{universal_out} is the output by Z3 \cite{demouraZ3EfficientSMT2008} (Z3 version 4.11.0).
Thus, we have that $\extpvsigzero^* \setminus (\bigcup_{i \in \range{1, 21}} \extpvsigzero^* \word[2]_i \extpvsigzero^*)$ is finite.

Additionally, \Cref{universal_28} is the SMT-LIB2 file for showing:
there is a \kl{word} $w$ of length $\wlen{w} \ge 28$ such that $w \in \extpvsigzero^* \setminus (\bigcup_{i \in \range{1, 21}} \extpvsigzero^* \word[2]_i \extpvsigzero^*)$.
\Cref{universal_28_out} is the output by Z3 \cite{demouraZ3EfficientSMT2008} (Z3 version 4.11.0).
For example, the following \kl{word} of length $28$ is in $\extpvsigzero^* \setminus (\bigcup_{i \in \range{1, 21}} \extpvsigzero^* \word[2]_i \extpvsigzero^*)$:
\[\iI\iD\cD\cD\re\cD\iI\cD\re\cD\cD\iD\re\cD\iD\re\cD\iD\re\cD\iD\re\cD\iD\re\cD\cD\iI.\]
(This \kl{word} can be obtained by uncommenting the last line of \Cref{universal_28}. See also \Cref{subsection: DFA}.)

\newpage

\begin{lstlisting}[ frame=h, caption=The SMT-LIB 2 file of \Cref{lemma: finite lang} for length $\ge 29$.,label=universal,captionpos=t,float,abovecaptionskip=-\medskipamount]
(set-logic QF_SLIA)
(define-const iI RegLan (str.to_re "1"))
(define-const iD RegLan (str.to_re "2"))
(define-const cD RegLan (str.to_re "3"))
(define-const re RegLan (str.to_re "4"))
(define-const Sigma RegLan (re.union iI iD cD re))
(define-const v01 RegLan (re.++ iI iI))
(define-const v02 RegLan (re.++ iD iD))
(define-const v03 RegLan (re.++ iI re))
(define-const v04 RegLan (re.++ re iI))
(define-const v05 RegLan (re.++ iD iI))
(define-const v06 RegLan (re.++ re iD))
(define-const v07 RegLan (re.++ re re))
(define-const v08 RegLan (re.++ cD iI iD))
(define-const v09 RegLan (re.++ iI cD iI))
(define-const v10 RegLan (re.++ iI cD iD))
(define-const v11 RegLan (re.++ iD cD iI))
(define-const v12 RegLan (re.++ cD iD cD))
(define-const v13 RegLan (re.++ cD cD cD))
(define-const v14 RegLan (re.++ cD cD iI cD))
(define-const v15 RegLan (re.++ cD iI cD cD))
(define-const v16 RegLan (re.++ iD re cD iI))
(define-const v17 RegLan (re.++ re cD re cD))
(define-const v18 RegLan (re.++ iD re cD cD iD))
(define-const v19 RegLan (re.++ re cD cD re cD))
(define-const v20 RegLan (re.++ re cD iD re cD iD re cD iD re cD iD re cD iD re))
(define-const v21 RegLan (re.++ re cD iI cD re cD iI cD re cD))

(define-const vunion RegLan (re.union v01 v02 v03 v04 v05 v06 v07 v08 v09 v10 v11 v12 v13 v14 v15 v16 v17 v18 v19 v20 v21))
(declare-const a String)
(assert (str.in_re a (re.inter (re.* Sigma) (re.comp (re.++ (re.* Sigma) vunion (re.* Sigma))))))
(assert (>= (str.len a) 29))
(check-sat)
\end{lstlisting}
\begin{lstlisting}[ frame=h, caption=The output by Z3 (Z3 version 4.11.0) for \Cref{universal}.,label=universal_out,captionpos=t,float,abovecaptionskip=-\medskipamount]
unsat
\end{lstlisting}

\lstinputlisting[ frame=h, caption=The SMT-LIB 2 file of \Cref{lemma: finite lang} for length $\ge 28$.,label=universal_28,captionpos=t,float,abovecaptionskip=-\medskipamount]{programs/universal\_28.smt2}
\lstinputlisting[ frame=h, caption=The output by Z3 (Z3 version 4.11.0) for \Cref{universal_28}.,label=universal_28_out,captionpos=t,float,abovecaptionskip=-\medskipamount]{programs/universal\_28.smt2.z3.out}

\clearpage
\subsection{Minimal DFA for \Cref{lemma: finite lang}} \label{subsection: DFA}
\Cref{fig:awesome_image} presents the minimal DFA of $(\bigcup_{i \in \range{1, 21}} \extpvsigzero^* \word[2]_i \extpvsigzero^*)$
(see online since the DFA is large; see also \cite{nakamuraRepositoryFiniteVariableOccurrence2023} for the dot file).
Its language is cofinite because this is \emph{acyclic} except the accepting state.
(The red colored edges denote the aforementioned word of length $28$ which is not accepted by the DFA.)
Thus, the minimal DFA also shows the cofiniteness of its language graphically.

\begin{sidewaysfigure}
    \centering

    \scalebox{0.168}{



     }
    \caption{The minimal DFA of $\bigcup_{i \in \range{1, 21}} \extpvsigzero^* \word[2]_i \extpvsigzero^*$.
        (See online).}
    \label{fig:awesome_image}
\end{sidewaysfigure}
  \fi

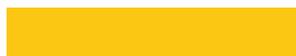
\end{document}